\documentclass[a4paper, 11pt]{amsart}
\usepackage{amsmath, amssymb, amsthm}	%ams environment
\usepackage{braket, stmaryrd, MnSymbol,ascmac,mathrsfs}
\usepackage{cite}

\calclayout

\newcommand{\pr}[1]{#1^{\prime}}
\newcommand{\ppr}[1]{#1^{\prime\prime}}
\newcommand{\pppr}[1]{#1^{\prime\prime\prime}}
\newcommand{\del}{\partial}
\newcommand{\vc}[1]{\boldsymbol #1}

\newcommand{\mfrak}[1]{\mathfrak{#1}}
\newcommand{\mcal}[1]{\mathcal{#1}}
\newcommand{\mbb}[1]{\mathbb{#1}}
\newcommand{\mrm}[1]{\mathrm{#1}}

\newcommand{\scr}[1]{\mathscr{#1}}
\newcommand{\what}[1]{\widehat{#1}}
\newcommand{\no}[1]{:\hspace{-3pt} #1\hspace{-3pt}:\hspace{3pt}}

\theoremstyle{plain}
\newtheorem{thm}{Theorem}[section]
\newtheorem{lem}[thm]{Lemma}
\newtheorem{prop}[thm]{Proposition}
\newtheorem{cor}[thm]{Corollary}
\theoremstyle{definition}
\newtheorem{defn}[thm]{Definition}
\theoremstyle{remark}

\makeatletter
    
    \@addtoreset{equation}{section}
\makeatother

\title[Local martingales associated with SLE with internal symmetry]{Local martingales associated with Schramm-Loewner evolutions  with internal symmetry}
\author{Shinji Koshida}
\address{Department of Basic Science, The University of Tokyo}
\email{koshida@vortex.c.u-tokyo.ac.jp}

\begin{document}

\begin{abstract}
We consider Schramm-Loewner evolutions (SLEs) with internal degrees of freedom that are associated with representations of
affine Lie algebras, following group theoretical formulation of SLEs.
We reconstruct the SLEs considered by Bettelheim {\it et al.} [Phys. Rev. Lett. {\bf 95}, 251601 (2005)] and Alekseev {\it et al.} [Lett. Math. Phys. {\bf 97}, 243-261 (2011)]
in correlation function formulation.
We also explicitly formulate stochastic differential equations on internal degrees of freedom for Heisenberg algebras and the affine $\mfrak{sl}_{2}$.
Our formulation enables us to find  several local martingales associated with SLEs with internal degrees of freedom
from computation on a representation of an affine Lie algebra.
Indeed, we formulate local martingales associated with SLEs with internal degrees of freedom described by Heisenberg algebras and the affine $\mfrak{sl}_{2}$.
We also find an affine $\mfrak{sl}_{2}$ symmetry of a space of SLE local martingales for the affine $\mfrak{sl}_{2}$.\end{abstract}

\maketitle

%introduction
\section{Introduction}
Growth processes have been proven to give frameworks that describe various equilibrium and non-equilibrium phenomena exhibited in nature.
Examples of such growth processes we consider in this paper are variants of Schramm-Loewner evolutions (SLEs),
which were introduced by Schramm\cite{Schramm2000} as the subsequent scaling limit of loop erased random walks and uniform spanning trees.
Actually, Schramm defined two types of SLEs, chordal and radial, but in this paper we only treat chordal SLEs  and simply call them SLEs.
SLE is the solution of the following stochastic Loewner equation
\begin{equation}
	\frac{d}{dt}g_{t}(z)=\frac{2}{g_{t}(z)-\sqrt{\kappa}B_{t}}
\end{equation}
on a formal power series $g_{t}(z)\in z+\mbb{C}[[z^{-1}]]$, with the initial condition $g_{0}(z)=z$.
Here $B_{t}$ is the standard Brownian motion with values in $\mbb{R}$ starting from the origin and $\kappa$ is a positive number.
The SLE specified by this number $\kappa$ is denoted as SLE$(\kappa)$.
Though we have regarded $g_{t}(z)$ as just a formal power series, it becomes a uniformization map of a simply connected domain.
Namely, for each realization of $g_{t}(z)$, we can take a subset $K_{t}\subset\mbb{H}$ called a hull such that
$g_{t}$ becomes a biholomorphic map $g_{t}:\mbb{H}\backslash K_{t}\to\mbb{H}$.
Moreover, for each realization, the family $\{K_{t}\}_{t\ge 0}$ of hulls parametrized by time is increasing, {\it i.e.}, if $t<s$, $K_{t}\subset K_{s}$ holds.
When we investigate an evolution of hulls in more detail, we find that it is governed by an evolution of the tip $\gamma_{t}$ in the upper half plane,
which is captured in the following manner.
At the initial time $t=0$, the uniformization map $g_{0}$ is the identity, which means that the hull $K_{0}$ is empty.
At a small time $t=t_{1}$, the corresponding hull $K_{t_{1}}$ is a slit in the upper half plane, one of whose end points is on the origin.
Then we name the other end point $\gamma_{t_{1}}$ and call it the tip at $t=t_{1}$.
For small time, the hull is simply  the trace of the tip,
but when time evolves further, the trace may touch itself or the real axis.
If such an event occurs, the area enclosed by the trace and the real axis is absorbed in the hull.
This is a  way of identifying the evolution of hulls with the evolution of the tip.
In this manner, SLEs give a probability measure on the space of curves in the upper half plane, which is called the SLE$(\kappa)$-measure.
The SLE$(\kappa)$-measure has been shown to describe an interface of clusters in several critical systems in two dimensions 
including the critical percolation \cite{Smirnov2001} and the Ising model at criticality  \cite{ChelkakDuminil-CopinHonglerKemppainenSmirnov2014}.
After their introduction, many aspects of SLEs have been clarified.\cite{Lawler2004,RohdeSchramm2005,LawlerSchrammWerner2001a,LawlerSchrammWerner2001b,LawlerSchrammWerner2002a,LawlerSchrammWerner2002b,Werner2003}.

There is another framework to investigate two-dimensional critical systems.
It is two-dimensional conformal field theory (CFT),\cite{BelavinPolyakovZamolodchikov1984}
which has been one of the most powerful tools in a wide variety of fields
ranging  from condensed matter physics to string theory, and in mathematics.
A milestone of CFT prediction on a critical system is Cardy's formula,\cite{Cardy1992}
which gives the crossing probability for the critical percolation in two dimensions from computation of a correlation function in CFT.
Cardy's formula was proven by Simirnov \cite{Smirnov2001} to be a theorem, while the derivation by Cardy has not been verified.

Since SLE and CFT are different frameworks that describe the same phenomena, they are expected to be connected to each other in some sense.
The connection between SLE and CFT has been studied under the name of SLE/CFT correspondence from various points of view.
Studies by Friedrich, Werner, Kalkkinen, and Kontsevich,\cite{FriedrichWerner2003, FriedrichKalkkinen2004,Friedrich2004,Kontsevich2003} 
proposed that the SLE$(\kappa)$-measure was constructed as a section of the determinant bundle over the moduli space of Riemann surfaces
based on observation of transformation of the correlation function of CFT under conditioning.
In  a more-recent approach by Dub\'{e}dat,\cite{Dubedat2015a,Dubedat2015b} the SLE$(\kappa)$-measure was constructed using a  localization technique,
and its partition function was identified with a highest weight vector of a representation of Virasoro algebra.
A significant development was  the {\it group theoretical} formulation of SLEs  by Bauer and Bernard,\cite{BauerBernard2002, BauerBernard2003a,BauerBernard2003b}
which proposed  an elegant way of constructing local martingales associated with SLEs (SLE local martingales for short)  from a representation of the Virasoro algebra.
We will review this formulation in Sect.\ref{sect:group_theoretical}.

The notion of SLE has been generalized to several directions along the SLE/CFT correspondence.
Examples include the notion of multiple SLEs\cite{BauerBernardKytola2005} and SLEs corresponding to logarithmic CFT,\cite{Rasmussen2004b,Moghimi-AraghiRajabpourRouhani2004} the $\mcal{N}=1$ superconformal algebra.\cite{Rasmussen2004a}

We note  that there are other directions of generalization of SLEs.
An example is the notion of SLE$(\kappa,\rho)$,\cite{LawlerSchrammWerner2003} which is obtained by replacing the Brownian motion in the stochastic Loewner equation
by a Bessel process.
CFT interpretation of SLE$(\kappa,\rho)$ was obtained by Cardy\cite{Cardy2006} and Kyt\"{o}l\"{a}.\cite{Kytola2006}
Several variants of SLEs associated with representation of the Virasoro algebra were unified by Kyt\"{o}l\"{a}.\cite{Kytola2007}

CFTs that are associated with representation theory of affine Lie algebras are known
as Wess-Zumino-Witten (WZW) theories.\cite{WessZumino1971, Witten1984, KnizhnikZamolodchikov1984}
SLEs corresponding to WZW theories have been considered
by Bettelheim {\it et al.}\cite{BettelheimGruzbergLudwigWiegmann2005} and Alekseev {\it et al.}\cite{AlekseevBytskoIzyurov2011} in correlation function formulation
and by Rasmussen\cite{Rasmussen2007} for the $\mfrak{sl}_{2}$ case and the present author\cite{SK2017} for simple Lie algebras in group theoretical formulation.
Note that the group theoretical formulation of SLEs corresponding to WZW theory first given by Rasmussen\cite{Rasmussen2007} did not contain the original SLE as a part,
and the  current author\cite{SK2017} presented an idea for improving it to recover the original SLE as the geometric part and the result given by correlation function formulation.
We will now review the approach in correlation function formulation\cite{BettelheimGruzbergLudwigWiegmann2005,AlekseevBytskoIzyurov2011}
of SLEs corresponding to WZW theory.
Let $\mfrak{g}$ be a finite-dimensional simple Lie algebra and $k\in\mbb{C}$ be a level.
They start from an object
\begin{equation}
	\label{eq:corr_func_martingale}
	\mcal{M}_{t}=\frac{\braket{\phi_{\Lambda}(z_{t})\phi_{\lambda_{1}}(z_{1})\cdots\phi_{\lambda_{N}}(z_{N})\phi_{\lambda_{1}^{\ast}}(\bar{z}_{1})\cdots\phi_{\lambda_{N}^{\ast}}(\bar{z}_{N})\phi_{\Lambda^{\ast}}(\infty)}^{\mfrak{g}}}{\braket{\phi_{\Lambda}(z_{t})\phi_{\Lambda^{\ast}}(\infty)}^{\mfrak{g}}}.
\end{equation}
Here $\phi_{\lambda}$ is the primary field corresponding to a weight $\lambda$,
with  the convention that $\lambda^{\ast}$ denotes the dual representation of $\lambda$.
The points $z_{1},\cdots, z_{N}$ are put on the upper half plane
and $z_{t}$ is the tip of the SLE slit defined by $z_{t}=\rho_{t}^{-1}(0)$, where $\rho_{t}(z)=g_{t}(z)+B_{t}$ satisfies $d\rho_{t}(z)=\frac{2dt}{\rho_{t}(z)}-dB_{t}$
with $B_{t}$ being the Brownian motion of variance $\kappa$.
The numerator of Eq.(\ref{eq:corr_func_martingale}) takes a value in the $\mfrak{g}$-invariant subspace of $L(\Lambda)\otimes L(\lambda_{1})\otimes\cdots\otimes L(\Lambda)^{\ast}$,
where $L(\lambda)$ is the irreducible representation of $\mfrak{g}$ of highest weight $\lambda$.
The denominator of Eq.(\ref{eq:corr_func_martingale}) takes  a value in the $\mfrak{g}$-invariant subspace of $L(\Lambda)\otimes L(\Lambda)^{\ast}$,
which is one dimensional due to Schur's Lemma.

Since a primary field of a WZW theory has internal degrees of freedom,
random evolution of a primary field involves ones along the internal degrees of freedom.
Studies by Ref. \cite{BettelheimGruzbergLudwigWiegmann2005} and Ref. \cite{AlekseevBytskoIzyurov2011}  proposed the following stochastic differential equation (SDE):
\begin{equation}
	\label{eq:sle_wzw_correlation_function}
	d\phi_{\lambda_{i}}(w_{i})=\mcal{G}_{i}\phi_{\lambda_{i}}(w_{i}),
\end{equation}
where $w_{i}=\rho_{t}(z_{i})$ and
\begin{equation}
	\mcal{G}_{i}=dt\left(\frac{2}{w_{i}}\del_{w_{i}}-\frac{\tau C_{i}}{2w_{i}^{2}}\right)-dB_{t}\del_{w_{i}}
	+\left(\frac{1}{w_{i}}\sum_{a}d\theta^{a}t_{i}^{a}+\frac{\tau}{2w_{i}^{2}}\sum_{a}t_{i}^{a}t_{i}^{a}dt\right).
\end{equation}
Here $\{t^{a}\}$ is a basis of $\mfrak{g}$ and $\{t^{a}_{i}\}$ are their representation matrices on $L(\lambda_{i})$.
Random processes $\theta^{a}$ are independent Brownian motions of variance $\tau$.
The number $C_{i}$ is the value of the Casimir on the representation $L(\lambda_{i})$.

The claim by Ref. \cite{BettelheimGruzbergLudwigWiegmann2005} and Ref. \cite{AlekseevBytskoIzyurov2011} is that the random process $\mcal{M}_{t}$
is a local martingale for a certain choice of $\kappa$ and $\tau$,
and Eq. (\ref{eq:sle_wzw_correlation_function}) is a generalization of the stochastic Loewner equation so as to correspond to a WZW theory.
Their formulation has been extended to multiple SLEs\cite{Sakai2013} and to coset WZW theories.\cite{Nazarov2012,Fukusumi2017}

The motivation in the present work is to better understand the previous studies\cite{BettelheimGruzbergLudwigWiegmann2005, AlekseevBytskoIzyurov2011}
on SLEs  corresponding to WZW theory.
In their {{formulations,}} the SDEs along internal degrees of freedom {{appear}} to be {\it ad hoc},
random processes along internal degrees of freedom are not constructed in a concrete way, and thus
local martingales that are associated with SLEs corresponding to WZW theory are hard to formulate.
These issues are addressed in this paper.
In particular, we will see that SDEs on internal degrees of freedom {{arise}} naturally in the group theoretical formulation.
We also construct {{a}} random process along internal degrees of freedom for Heisenberg algebras and the affine $\mfrak{sl}_{2}$,
and {{formulate}} several local martingales associated with them.

This paper is organized as follows.
In Sect. \ref{sect:group_theoretical}, we review the group theoretical formulation of {{SLEs originally proposed}} by Bauer and Bernard.\cite{BauerBernard2002, BauerBernard2003a,BauerBernard2003b}
In Sect. \ref{sect:rep_aff_alg}, we recall the notion of affine Lie algebras associated with {{finite-dimensional}} Lie algebras that are simple or commutative and their representation theory.
In Sect. \ref{sect:internal_symmetry}, we introduce an infinite-dimensional Lie group,
which becomes the target space of random processes generating SLEs corresponding to representations of affine Lie algebras.
In Sect. \ref{sect:random_process}, we construct a random process on the infinite-dimensional Lie group 
assuming existence of an annihilating operator of a highest weight vector.
We also {{formulate}} SDEs on internal degrees of freedom in the case when the underlying Lie algebra is commutative and $\mfrak{sl}_{2}$.
In Sect. \ref{sect:annihilator}, we discuss an annihilating operator of a highest weight vector, {{the existence of which}} is assumed in Sect. \ref{sect:random_process}.
In Sect. \ref{sect:martingale}, as an application of {{the}} construction of SDEs in Sect. \ref{sect:random_process},
we compute several local martingales associated with the solutions.
In Sect. \ref{sect:affine_symmetry}, we clarify {{the}} $\what{\mfrak{sl}}_{2}$-module structure on a space of SLE local martingales for $\what{\mfrak{sl}}_{2}$.
{{Then we present some conclusions.}}
In Appendix \ref{sect:app_voa}, we recall the notion of vertex operator {{algebra (VOA)}}, which is useful in this paper.
In Appendix \ref{sect:app_ito_lie_group}, we {{review}} an Ito process on a Lie group.
Appendix \ref{sect:app_SDE} contains computational details that are referred to in Sect.\ref{sect:random_process}.
In Appendix \ref{sect:app_operator_X}, we show {{a}} detailed derivation of operators that define {{the}} action of $\what{\mfrak{sl}}_{2}$ on a space of local martingales
referred to in Sect.\ref{sect:affine_symmetry}.

%group_theoretical
\section{Group theoretical formulation of SLEs}
\label{sect:group_theoretical}
In this section, we recall the group theoretical formulation of {{SLEs}} corresponding to the Virasoro algebra {{originally proposed}} by Bauer and Bernard.\cite{BauerBernard2002, BauerBernard2003a,BauerBernard2003b}
The main purpose of this section is to introduce the {{infinite-dimensional}} Lie group $\mrm{Aut}_{+}\mcal{O}$ and a random process on it.

\subsection{Virasoro algebra and its representations}
Virasoro algebra is an {{infinite-dimensional}} Lie algebra $\mrm{Vir}=\bigoplus_{n\in\mbb{Z}}\mbb{C}L_{n}\oplus\mbb{C}C$
with Lie brackets defined by
\begin{align}
	[L_{m},L_{n}]&=(m-n)L_{m+n}+\frac{m^{3}-m}{12}\delta_{m+n,0}C, \\
	[C,\mrm{Vir}]&=\{0\}.
\end{align}

We only consider highest weight representations of the Virasoro algebra that are constructed in the following manner.
Let us decompose the Virasoro algebra into subalgebras $\mrm{Vir}=\mrm{Vir}_{>0}\oplus\mrm{Vir}_{0}\oplus \mrm{Vir}_{<0}$,
where $\mrm{Vir}_{0}=\mbb{C}L_{0}\oplus\mbb{C}C$ and $\mrm{Vir}_{\gtrless 0}=\bigoplus_{\pm n>0}\mbb{C}L_{n}$.
We also set $\mrm{Vir}_{\ge 0}=\mrm{Vir}_{0}\oplus\mrm{Vir}_{>0}$.
For a pair $(c,h)\in\mbb{C}^{2}$, let $\mbb{C}_{(c,h)}=\mbb{C}{\bf 1}_{(c,h)}$ be a {{one-dimensional}} representation of $\mrm{Vir}_{\ge 0}$
on which $C$ and $L_{0}$ act as multiplication by $c$ and $h$, respectively.
The highest weight Verma module $M(c,h)$ of highest weight $(c,h)$ is defined by induction $M(c,h)=U(\mrm{Vir})\otimes_{U(\mrm{Vir}_{\ge 0})}\mbb{C}_{(c,h)}$,
which is isomorphic to $U(\mrm{Vir}_{<0})\otimes \mbb{C}_{(c,h)}$ as a vector space or a $\mrm{Vir}_{<0}$-module.
The numbers $c$ and $h$ in the highest weight are called the central charge and the conformal weight of the highest weight Verma module $M(c,h)$, respectively.
Since we will only treat highest weight representations, we call a highest weight Verma module simply a Verma module.
The highest weight vector $1\otimes {\bf 1}_{(c,h)}$ is denoted by $\ket{c,h}$.
It is clear by construction that a Verma module $M(c,h)$ decomposes into {{the}} direct sum of eigenspaces of $L_{0}$ so that $M(c,h)=\bigoplus_{n\in\mbb{Z}_{\ge 0}}M(c,h)_{h+n}$,
where we have defined $M(c,h)_{\lambda}=\{v\in M(c,h)|L_{0}v=\lambda v\}$ for $\lambda\in\mbb{C}$.

For a generic {{highest}} weight $(c,h)$, the corresponding Verma module is irreducible,
but for a specific highest weight, it is not.
Then we denote the irreducible quotient of the Verma module by $L(c,h)$,
and call an element in $J(c,h):=\ker (M(c,h)\twoheadrightarrow L(c,h))$ a null vector.

Among other irreducible modules, that of the highest weight $(c,0)$ denoted by $L(c,0)$ above has {{a special feature, which is}} that it carries {{the}} structure of a {{VOA}}.
We simply denote this VOA by $L_{c}$ and call it the Virasoro VOA of central charge $c$.
An exposition of {{VOA}} structure on $L_{c}$ is presented in Appendix \ref{sect:app_voa},
and we shall sketch the argument here.
The vacuum vector is the highest weight vector $\ket{0}=\ket{c,0}$,
and it is generated by a conformal vector $L_{-2}\ket{0}$ that is transferred to the Virasoro field $L(z)=\sum_{n\in\mbb{Z}}L_{n}z^{-n-2}$
under the state-field correspondence map.
Simple modules over the Virasoro VOA $L_{c}$ are realized as highest weight irreducible representations of the same central charge.
A {{nondegenerate}} bilinear form $\braket{\cdot|\cdot}$ on an $L_{c}$-module $M$ is invariant if it satisfies
\begin{equation}
	\label{eq:invariant_form}
	\braket{Y(a,z)u|v}=\braket{u|Y(e^{zL_{1}}(-z^{-2})^{L_{0}}a,z^{-1})v}
\end{equation}
for $a\in L_{c}$ and $u, v \in M$.
This condition is rephased as $\braket{L_{n}u|v}=\braket{u|L_{-n}v}$ and $\braket{Cu|v}=\braket{u|Cv}$ for $u, v\in M$,
which specify a bilinear form $\braket{\cdot|\cdot}$ on $M$.
It is {{well known}} that such a bilinear form uniquely exists under the normalization $\braket{c,h|c,h}=1$.

\subsection{Conformal transformation}
Here we review how to implement a conformal transformation as an operator on a VOA or its module following {{Frenkel and Ben-Zvi}}.\cite{FrenkelBen-Zvi2004}
Let $\mcal{O}=\mbb{C}[[w]]=\varprojlim \mbb{C}[w]/(w^{n})$ be a complete topological $\mbb{C}$-algebra and $D=\mrm{Spec}\mcal{O}$ be the formal {{disk}}.
A continuous automorphism $\rho$ of $\mcal{O}$ is identified with the image of the topological generator $w$ of $\mcal{O}$
by the same automorphism $\rho$.
Under this identification, the group $\mrm{Aut}\mcal{O}$ of continuous {{automorphisms}} of $\mcal{O}$ is realized as
\begin{equation}
	\label{eq:aut_at_infty}
	\mrm{Aut}\mcal{O}\simeq \{a_{1}w+a_{2}w^{2}+\cdots|a_{1}\in\mbb{C}^{\times},\ a_{i}\in\mbb{C},\ i\ge 2\}.
\end{equation}
Indeed, a nonzero constant term is prohibited to preserve the algebra $\mcal{O}$, and $a_{1}\neq 0$ is required for the existence of the inverse.
The group law is defined by $(\rho\ast\mu)(w)=\mu(\rho(w))$ for $\rho,\ \mu\in\mrm{Aut}\mcal{O}$.
The purpose of this subsection is to define a representation of this group on a {{VOA}} or its modules,
which is significant in application to the theory of {{SLEs}}.

It is shown that the Lie algebra of $\mrm{Aut}\mcal{O}$ is one of vector fields $\mrm{Der}_{0}\mcal{O}=w\mbb{C}[[w]]\del_{w}$.
The same Lie algebra is also constructed as a completion of a Lie subalgebra $\mrm{Vir}_{\ge 0}=\bigoplus_{n=0}^{\infty}\mbb{C}L_{n}$ of the Virasoro algebra.
Since a subalgebra $\mrm{Vir}_{\ge m}=\bigoplus_{n\ge m}\mbb{C}L_{n}$ in $\mrm{Vir}_{\ge 0}$ is an ideal,
the quotient $\mrm{Vir}_{\ge 0}/\mrm{Vir}_{\ge m}$ carries a Lie algebra structure{{;}}
moreover, we have a family of projections $\mrm{Vir}_{\ge 0}/\mrm{Vir}_{\ge m}\to \mrm{Vir}_{\ge 0}/\mrm{Vir}_{\ge n}$ for $m>n$.
The projective limit $\varprojlim \mrm{Vir}_{\ge 0}/\mrm{Vir}_{\ge m}$ of this projective system of Lie algebras 
is the desired Lie algebra $\mrm{Der}_{0}\mcal{O}$.
Since, for an arbitrary vector $v$ in a {{VOA}}, $V$ or its module $M$,
$L_{n}v=0$ for $n\gg0$, {{so}} we have a well-defined action of $\mrm{Der}_{0}\mcal{O}$ on $V$ and $M$.

There is a significant subgroup $\mrm{Aut}_{+}\mcal{O}$ of $\mrm{Aut}\mcal{O}$ that is described as
$\mrm{Aut}_{+}\mcal{O}\simeq \{w+a_{2}w^{2}+\cdots|a_{i}\in\mbb{C},\ i\ge 2\}$.
It is shown that the Lie algebra of this subgroup is $\mrm{Der}_{+}\mcal{O}=w^{2}\mbb{C}[[w]]\del_{w}${{, which}} is a Lie subalgebra of $\mrm{Der}_{0}\mcal{O}$.

We shall exponentiate the action of the Lie algebra $\mrm{Der}_{0}\mcal{O}$ to the action of the Lie group $\mrm{Aut}\mcal{O}$.
This is possible if {{the}} $L_{n}$ for $n>1$ act locally nilpotently and $L_{0}$ is diagonalizable with integer eigenvalues,
{{the}} former of which automatically holds for a highest weight representation,
and {{the}} latter of which is true if the conformal weight of the highest weight is an integer.
On such a highest weight representation of the Virasoro algebra, we construct the linear operator $R(\rho)$ for $\rho\in\mrm{Aut}\mcal{O}$
that defines a representation of $\mrm{Aut}\mcal{O}$.
For an automorphism $\rho\in \mrm{Aut}\mcal{O}$,
we uniquely find $v_{i}$, $i\ge 0$, such that
\begin{equation}
	\rho(w)=\exp\left(\sum_{i>0}v_{i}w^{i+1}\del_{w}\right)v_{0}^{w\del_{w}}\cdot w.
\end{equation}
Here the exponentiation of the Euler vector field is defined by $v_{0}^{w\del_{w}}\cdot w=v_{0}$.
The above expression of $\rho$ is {{a}} specification of its action on $\mcal{K}=\mbb{C}((w))$ defined by $(\rho. F)(w)=f(\rho(w))$ for $F(w)\in \mcal{K}$,
where the group law of invertible operators on $\mcal{K}$ is defined by composition.
The first few {{values}} of $v_{i}$ for a given $\rho$ are computed by comparing coefficients of each powers of $w$ so that
\begin{align*}
	v_{0}&=\pr{\rho}(0), &
	v_{1}&=\frac{1}{2}\frac{\ppr{\rho}(0)}{\pr{\rho}(0)}, &
	v_{2}&=\frac{1}{6}\frac{\pppr{\rho}(0)}{\pr{\rho}(0)}-\frac{1}{4}\left(\frac{\ppr{\rho}(0)}{\pr{\rho}(0)}\right)^{2}, &\cdots.
\end{align*}
Let $V$ be a VOA.
Then for an automorphism $\rho\in\mrm{Aut}\mcal{O}$, the following operator is {{well defined}} in $\mrm{End}(V)${{:}}
\begin{equation}
	R(\rho)=\exp\left(-\sum_{i>0}v_{i}L_{i}\right)v_{0}^{-L_{0}},
\end{equation}
and satisfies $R(\rho)R(\mu)=R(\rho\ast\mu)$.
In {{the case when}} $\rho\in\mrm{Aut}_{+}\mcal{O}$, we have $v_{0}=1$, which means that $R(\rho)$ can also be regarded as an operator on a $V$-module.

We investigate the behavior of a field $Y(A,z)$ on a {{VOA}} $V$ under the adjoint action by $R(\rho)$.
Let $L(z)=\sum_{n\in\mbb{Z}}L_{n}z^{-n-2}$ be the Virasoro field{{. Then,}}
\begin{equation}
	[L(z),Y(A,w)]=\sum_{m\ge -1}Y(L_{m}A,w)\del_{w}^{(m+1)}\delta(z-w),
\end{equation}
which implies
\begin{equation}
	[L_{n},Y(A,w)]=\sum_{m\ge -1}\binom{n+1}{m+1}Y(L_{m}A,w)w^{n-m}.
\end{equation}
For ${\bf v}=-\sum_{n\in\mbb{Z}}v_{n}L_{n}$ such that $v_{n}=0$ for $n\ll 0$, {{then}}
\begin{equation}
	[{\bf v},Y(A,w)]=-\sum_{m\ge -1}\left(\del_{w}^{(m+1)}v(w)\right)Y(L_{m}A,w),
\end{equation}
where $v(w)=\sum_{n\in\mbb{Z}}v_{n}w^{n+1}$.

\begin{prop}
For $A\in V$ and $\rho\in\mrm{Aut}\mcal{O}$,
\begin{equation}
	Y(A,w)=R(\rho)Y(R(\rho_{w})^{-1}A,\rho(w))R(\rho)^{-1}.
\end{equation}
Here $\rho_{w}(t)=\rho(w+t)-\rho(w)$.
\end{prop}
\begin{proof}
We denote by $\mrm{Fie}(V)$ the space of fields on $V$.
The state field correspondence map $Y(-,w)$ is regarded as an element in $\mrm{Hom}(V,\mrm{Fie}(V))$.
For an automorphism $\rho\in\mrm{Aut}\mcal{O}$, we define an endomorphism $T_{\rho}$ on $\mrm{Hom}(V,\mrm{Fie}(V))$ by
\begin{equation}
	(T_{\rho}\cdot X)(A,w):=R(\rho)X(R(\rho_{w})^{-1}A,\rho(w))R(\rho)^{-1}
\end{equation}
for $X\in \mrm{Hom}(V,\mrm{Fie}(V))$ and $A\in V$.
Then this assignment $\rho\mapsto T_{\rho}$ is a group homomorphism.
Indeed,
\begin{align*}
	&(T_{\rho}\cdot(T_{\mu}\cdot X))(A,w) \\
	&=R(\rho)(T_{\mu}\cdot X)(R(\rho_{w})^{-1}A,\rho(w))R(\rho)^{-1} \\
	&=R(\rho)R(\mu)X(R(\mu_{\rho(w)})^{-1}R(\rho_{w})^{-1}A,\mu(\rho(w)))R(\mu)^{-1}R(\rho)^{-1}.
\end{align*}
{{Note}} that
\begin{align*}
	(\rho_{w}\ast \mu_{\rho(w)})(t)
	&=\mu_{\rho(w)}(\rho_{w}(t))=\mu(\rho(w)+\rho_{w}(t))-\mu(\rho(w)) \\
	&=\mu(\rho(w)+\rho(w+t)-\rho(w))-\mu(\rho(w)) \\
	&=(\rho\ast\mu)_{w}(t)
\end{align*}
to obtain
\begin{equation}
	(T_{\rho}\cdot(T_{\mu}\cdot X))(A,w)=(T_{\rho\ast\mu}\cdot X)(A,w).
\end{equation}

Since the exponential map $\mrm{Der}_{0}\mcal{O}\to \mrm{Aut}\mcal{O}$ is surjective,
we can assume $\rho$ to be infinitesimal.
For an infinitesimal transformation $\rho(w)=w+\epsilon v(w)+o(\epsilon)$ with $v(w)=\sum_{n\ge 0}v_{n}w^{n+1}$,
\begin{equation}
	R(\rho)=\mrm{Id}+\epsilon {\bf v}+o(\epsilon),
\end{equation}
where ${\bf v}=-\sum_{n\ge 0}v_{n}L_{n}$.
The associated transformation $\rho_{w}(t)$ is approximated {{up to a}} linear order of $\epsilon$ by
\begin{align*}
	\rho_{w}(t)
		&=\rho(w+t)-\rho(w)=w+t+\epsilon v(w+t)-w-\epsilon v(w)+o(\epsilon) \\
		&=t+\epsilon\sum_{m\ge 0}\del^{(m+1)}v(w)t^{m+1}+o(\epsilon).
\end{align*}
Thus $R(\rho_{w})^{-1}$ becomes
\begin{equation}
	R(\rho_{w})^{-1}=\mrm{Id}+\epsilon \sum_{n\ge 0}\del^{(n+1)}v(w)L_{n}+o(\epsilon).
\end{equation}
We now show that the state-field correspondence map $Y(-,w)$ is fixed under the action of $T_{\rho}$
up to {{a}} linear order of $\epsilon${{:}}
\begin{align*}
	&(T_{\rho}\cdot Y)(A,w) \\
	&=(\mrm{Id}+\epsilon {\bf v})Y\left(\left(\mrm{Id}+\epsilon\sum_{n\ge 0}\del^{(n+1)}v(w)L_{n}\right)A,w+\epsilon+v(w)\right)(\mrm{Id}-\epsilon{\bf v}) \\
	&=Y(A,w)+\epsilon\left([{\bf v},Y(A,w)]+v(w)\del Y(A,w)+\sum_{n\ge 0}\del^{(n+1)}v(w)Y(L_{n}A,w)\right) \\
	&=Y(A,w).
\end{align*}
\end{proof}

\begin{cor}
Let $A\in V$ be a primary vector of conformal weight $h$, i.e., it satisfies $L_{n}A=0$ for $n>0$ and $L_{0}A=hA$.
For an automorphism $\rho\in\mrm{Aut}\mcal{O}$,
\begin{equation}
	Y(A,w)=R(\rho)Y(A,\rho(w))R(\rho)^{-1}(\pr{\rho}(w))^{h}.
\end{equation}
\end{cor}
\begin{proof}
For a primary vector $A$ of conformal weight $h$,
the {{one-dimensional}} space $\mbb{C}A$ is preserved by the operator $R(\rho_{w})$,
where $R(\rho_{w})$ is given by
\begin{equation}
	R(\rho_{w})=\exp\left(-\sum_{j>0}v_{j}(w)L_{j}\right)v_{0}(w)^{-L_{0}}
\end{equation}
with $v_{j}(w)$ being chosen so that 
\begin{equation}
	\rho_{w}(t)=\exp\left(\sum_{j>0}v_{j}(w)t^{j+1}\del_{t}\right)v_{0}(w)^{t\del_t}\cdot t.
\end{equation}
Since $A$ is primary, the nontrivial effect comes from the action by $L_{0}$, thus we have $R(\rho_{w})A=v_{0}(w)^{-h}A$,
where $v_{0}(w)$ is computed as $v_{0}(w)=\del_{t}\rho_{w}(t=0)=\pr{\rho}(w)$,
which implies that $R(\rho_{w})^{-1}A=(\pr{\rho}(w))^{h}A$.
\end{proof}

One {{important field that is}} not primary is the Virasoro field $L(w)=Y(L_{-2}\ket{0},w)$,
which transforms as follows.
\begin{prop}
	Let $L(w)$ be the Virasoro field{{:}}
	\begin{equation}
	L(w)=R(\rho)L(\rho(w))R(\rho)^{-1}(\pr{\rho}(w))^{2}+\frac{c}{12}(S\rho)(w).
	\end{equation}
	Here $c\in\mbb{C}$ is the central charge and $(S\rho)(w)$ is the Schwarzian derivative defined by
	\begin{equation}
		(S\rho)(w)=\frac{\pppr{\rho}(w)}{\pr{\rho}(w)}-\frac{3}{2}\left(\frac{\ppr{\rho}(w)}{\pr{\rho}(w)}\right)^{2}.
	\end{equation}
\end{prop}
\begin{proof}
It is clear that the space $\mbb{C}L_{-2}\ket{0}\oplus\mbb{C}\ket{0}$ is preserved by the operator $R(\rho_{w})$,
thus we first compute the inverse $R(\rho_{w})^{-1}$ on this space.
Let $v_{j}(w)\in \mbb{C}[[w]]$ be chosen so that
\begin{equation}
	\rho_{w}(t)=\exp\left(\sum_{j>0}v_{j}(w)t^{j+1}\del_{t}\right)v_{0}(w)^{t\del_{t}}\cdot t,
\end{equation}
then $R(\rho_{w})$ {{can be}} expressed as
\begin{equation}
	R(\rho_{w})=\exp\left(-\sum_{j>0}v_{j}(w)L_{j}\right)v_{0}(w)^{-L_{0}}.
\end{equation}
The matrix form of this operator on $\mbb{C}L_{-2}\ket{0}\oplus\mbb{C}\ket{0}$ is expressed in this basis
\begin{equation}
	R(\rho_{z})=\left(
	\begin{array}{cc}
		v_{0}(w)^{-2} & 0 \\
		-\frac{c}{2}v_{0}(w)^{-2}v_{2}(w) & 1
	\end{array}
	\right),
\end{equation}
and its inverse is
\begin{equation}
	R(\rho_{w})^{-1}=\left(
	\begin{array}{cc}
		v_{0}(w)^{2} & 0 \\
		\frac{c}{2}v_{2}(w) & 1
	\end{array}
	\right)=\left(
	\begin{array}{cc}
		(\pr{\rho}(w))^{2} & 0 \\
		\frac{c}{12}(S\rho)(w) & 1
	\end{array}
	\right),
\end{equation}
which implies the desired result.
\end{proof}

In application to the theory of SLE, we regard the formal {{disk}} introduced here as the formal neighborhood at infinity,
and have to reformulate {{all the components so that they are}} associated with the coordinate $z=\frac{1}{w}$ at $0$.
While an automorphism $\rho$ sends $w$ to $\rho(w)=a_{1}w+a_{2}w^{2}+\cdots$, the same automorphism sends $z$ to $1/\rho(1/z)$.
If we expand the image in $z\mbb{C}[[z^{-1}]]$, we can also identify the group $\mrm{Aut}\mcal{O}$ with
\begin{equation}
	\label{eq:aut_at_zero}
	\mrm{Aut}\mcal{O}\simeq\{b_{1}z+b_{0}+b_{-1}z^{-1}+\cdots|b_{1}\in\mbb{C}^{\times},\ b_{i}\in \mbb{C},\ i\le 0\}
\end{equation}
The infinite series in $z\mbb{C}[[z^{-1}]]$ that is identified with an automorphism $\rho$ will be denoted by $\rho(z)$.
In the following, we regard formal variables $z$ and $w$ as formal {{coordinates}} at $0$ and infinity, respectively,
and $\rho(z)$ and $\rho(w)$ as infinite series identified with an automorphism $\rho$ via Eq.(\ref{eq:aut_at_infty}) and Eq.(\ref{eq:aut_at_zero}), respectively.

Under realization Eq.(\ref{eq:aut_at_zero}) of the group $\mrm{Aut}\mcal{O}$, its subgroup $\mrm{Aut}_{+}\mcal{O}$ consists of formal series
$z+b_{0}+b_{-1}z^{-1}+\cdots$ with $b_{i}\in \mbb{C}$ for $i\le 0$,
and Lie algebras are realized as $\mrm{Der}_{+}\mcal{O}=\mbb{C}[[z^{-1}]]\del_{z}$ and $\mrm{Der}_{0}\mcal{O}=z\mbb{C}[[z^{-1}]]\del_{z}$.

Since the Lie algebra $\mrm{Der}_{0}\mcal{O}=z\mbb{C}[[z^{-1}]]\del_{z}$ consists of vector fields{{, the coefficients of which}} are Laurent series in $z^{-1}$,
it cannot act on a VOA $V$ or its module $M$ by assignment $-z^{n+1}\del_{z}\to L_{n}$ for $n\le 0$.
Nevertheless, we can define {{well-defined}} operators that represent the Lie algebra $\mrm{Der}_{0}\mcal{O}$ on the completion of the vector space.
Let $M=\bigoplus_{n\in\mbb{Z}}M_{n}$ be the $\mbb{Z}$-gradation of a $V$-module $M$.
Then we define its formal completion by $\overline{M}=\prod_{n\in\mbb{Z}}M_{n}$.
Recall that $M_{n}=0$ for sufficiently small $n$.
Moreover this action of $\mrm{Der}_{0}\mcal{O}$ is exponentiated as a representation of $\mrm{Aut}\mcal{O}$ on $\overline{V}$,
and a representation of its subgroup $\mrm{Aut}_{+}\mcal{O}$ on $\overline{M}$.

For a given $\rho\in\mrm{Aut}\mcal{O}$, we can uniquely find numbers $v_{i}$ ($i\le 0$) that satisfy
\begin{equation}
	\exp\left(\sum_{j<0}v_{j}z^{j+1}\del_{z}\right)v_{0}^{z\del_{z}}\cdot z=\rho(z).
\end{equation}
Then the operator $Q(\rho)$ defined by
\begin{equation}
	Q(\rho)=\exp\left(-\sum_{j<0}v_{j}L_{j}\right)v_{0}^{-L_{0}}
\end{equation}
is a well-defined one on $\overline{V}$ and a representation of $\mrm{Aut}\mcal{O}$ {{can be defined}}.
Indeed, the part $v_{0}^{-L_{0}}$ behaves as multiplication by $v_{0}^{-n}$ when restricted on $V_{n}$ ,
and $L_{j}$ with $j<0$ strictly raises the degree, while the $\mbb{Z}$-gradation on $V$ is bounded from below.

We investigate the covariance property of a field $Y(A,z)$ under the adjoint action by $Q(\rho)$.
For $v(z)=\sum_{n\in\mbb{Z}}v_{n}z^{n+1}\in\mbb{C}((z^{-1}))$,
\begin{equation}
	[{\bf v},Y(A,z)]=\sum_{m\ge -1}\del^{(m+1)}v(z)Y(L_{m}A,z),
\end{equation}
with ${\bf v}=-\sum_{n\in\mbb{Z}}v_{n}L_{n}$, but here the both sides belong to $\mrm{End}(\overline{V})[z,z^{-1}]$.

\begin{prop}
\label{prop:transformation_zero}
For $A\in V$ and $\rho\in\mrm{Aut}\mcal{O}$,
\begin{equation}
	Y(A,z)=Q(\rho)Y(R(\rho_{z})^{-1}A,\rho(z))Q(\rho)^{-1}.
\end{equation}
\end{prop}

On a $V$-module on which eigenvalues of $L_{0}$ are not integers, the whole group $\mrm{Aut}\mcal{O}$ cannot act,
while its subgroup $\mrm{Aut}_{+}\mcal{O}$ can act.
In application to {{SLEs}}, this subgroup is sufficient since each realization of the SLE is always normalized
so that its expansion around infinity begins from $z$ with the coefficient unity.

For an operator $T$ on a VOA $V$, we define its adjoint operator $T^{\ast}$ by the property that
$\braket{Tu|v}=\braket{u|T^{\ast}v}$ for $u, v\in V$.
In this terminology, the operator $Q(\rho)$ defined above is the inverse of the adjoint operator of $R(\rho)$,
while $Q(\rho)$ is not an operator on a VOA but on its formal completion.

\subsection{Appearance of SDEs}
A fundamental object in the group theoretical formulation of {{SLEs}} is a random process $\rho_{t}$ on the {{infinite-dimensional}} Lie group $\mrm{Aut}_{+}\mcal{O}$.
A random process on a Lie group induces one on the space of operators on a representation space.
Let us take $(\gamma, \mcal{K}=\mbb{C}((z^{-1})))$ as a representation of $\mrm{Aut}_{+}(\mcal{O})$ defined by
$(\gamma(\rho)F)(z)=F(\rho(z))$.
Following {{the}} description of a random process on a Lie group presented in Appendix \ref{sect:app_ito_lie_group},
we assume that the induced random process on $\mrm{Aut}\mcal{K}$ satisfies the SDE
\begin{equation}
	\gamma(\rho_{t})^{-1} d\gamma(\rho_{t})=\left(2z^{-1}\del_{z}+\frac{\kappa}{2}\del_{z}^{2}\right)dt-\del_{z}dB_{t}
\end{equation}
under the initial condition $\gamma(\rho_{0})=\mrm{Id}$.
Here $B_{t}$ is the $\mbb{R}$-valued Brownian motion of variance $\kappa$ that starts from the origin.
Then we observe that $\gamma(\rho_{t})z=\rho_{t}(z)$ satisfies the SDE
\begin{equation}
	d\rho_{t}(z)=\frac{2}{\rho_{t}(z)}dt-dB_{t}
\end{equation}
under the initial condition $\rho_{0}(z)=z$.
If we introduce $g_{t}(z)=\rho_{t}(z)+B_{t}$, we find that $g_{t}(z)$ satisfies the stochastic Loewner equation
\begin{equation}
	\label{eq:sle_equation}
	\frac{d}{dt}g_{t}(z)=\frac{2}{g_{t}(z)-B_{t}}.
\end{equation}
Moreover, since $B_{0}=0$, we have $g_{0}(z)=z$.
Thus $g_{t}(z)$ is identified with the SLE$(\kappa)$.

We have just derived the stochastic Loewner equation from a random process on the Lie group $\mrm{Aut}_{+}\mcal{O}$.
This manner of formulation enables us to obtain several local martingales associated with SLE.
Let us consider the object $Q(\rho_{t})\ket{c,h}$, which is regarded as a random process on $\overline{L(c,h)}$,
of which the increment is
\begin{equation}
	d(Q(\rho_{t})\ket{c,h})=Q(\rho_{t})\left(\left(-2L_{-2}+\frac{\kappa}{2}L_{-1}^{2}\right)\ket{c,h}dt+L_{-1}\ket{c,h}dB_{t}\right).
\end{equation}
Thus if the vector $\chi=\left(-2L_{-2}+\frac{\kappa}{2}L_{-1}^{2}\right)\ket{c,h}$ is a null vector in the Verma module $M(c,h)$,
the random process $Q(\rho_{t})\ket{c,h}$ is a local martingale.
Notice that $\chi$ is a null vector if and only if it is a singular vector, conditions for which are that we have
$c=1-\frac{3(\kappa-4)^{2}}{2\kappa}$ and $h=\frac{6-\kappa}{2\kappa}$.
Thus for such a choice of $(c,h)$, the random process $Q(\rho_{t})\ket{c,h}$ in $\overline{L(c,h)}$ is a local martingale,
and produces several local martingales associated with SLE.
An example is given by $\braket{c,h|L(z)Q(\rho_{t})|c,h}$, where $L(z)$ is the Virasoro field on $L(c,h)$.
From Prop. \ref{prop:transformation_zero} and the fact that the dual of the highest weight vector $\bra{c,h}$ is invariant under the right action by $Q(\rho)$, we find that
\begin{equation}
	\braket{c,h|L(z)Q(\rho_{t})|c,h}=h\left(\frac{\pr{\rho}_{t}(z)}{\rho_{t}(z)}\right)^{2}+\frac{c}{12}(S\rho_{t})(z)
\end{equation}
is a local martingale.
We can show that such a quantity is indeed a local martingale by a standard Ito calculus,
but the group theoretical formulation of SLE in the sense of Bauer and Bernard\cite{BauerBernard2002, BauerBernard2003a,BauerBernard2003b} further clarifies its theoretical origin.

Since the solution $g_{t}$ of the stochastic Loewner equation is also described as $g_{t}(z)=(\rho_{t}\ast (z+ B_{t}))(z)$,
the operator $Q(g_{t})$ corresponding to $g_{t}$ is written as $Q(g_{t})=Q(\rho_{t})e^{-B_{t}L_{-1}}$.
Let $\mcal{Y}(-,z)$ be an intertwining operator of type $\binom{L(c,h)}{L(c,h)\ \ L_{c}}$,
then $\mcal{Y}(\ket{c,h},z)$ is a primary field, which is applied to the vacuum vector $\ket{0}$ to {{yield}}
$\mcal{Y}(\ket{c,h},z)\ket{0}=e^{zL_{-1}}\ket{c,h}$.
If we are allowed to substitute the Brownian motion $B_{t}$ in the formal variable $z$, {{then}}
\begin{equation}
	Q(g_{t})\mcal{Y}(\ket{c,h},B_{t})\ket{0}=Q(\rho_{t})\ket{c,h},
\end{equation}
which is a local martingale for a certain choice of $(c,h)$ depending on $\kappa$.
The {{left-hand side is}} a convenient form of the same local martingale
in revealing a Virasoro module structure on a space of SLE local martingales.\cite{Kytola2007}

The origin of the {{infinite-dimensional}} Lie group $\mrm{Aut}\mcal{O}$ in CFT was a seminal work,\cite{KawamotoNamikawaTsuchiyaYamada1988}
in which the group $\mrm{Aut}\mcal{O}$ appeared as a part of the fiber of the fiber bundle $\what{\mfrak{M}}_{g,n}\to\mfrak{M}_{g,n}$,
where $\mfrak{M}_{g,n}$ is the moduli space of Riemann surfaces of genus $g$ and with $n$ punctures,
and $\what{\mfrak{M}}_{g,n}$ is the moduli space decorated by local coordinates at punctures.
One can put on this {{infinite-dimensional}} Lie group a line bundle, the sheaf of sections of which admits an action of the Virasoro algebra.\cite{KirillovYuriev1988}
{{This}} action essentially gives rise to the Virasoro action on the space of SLE local martingales.\cite{BauerBernard2003b}
These subjects on SLE {{were}} developed and unified {{by Friedrich}}.\cite{Friedrich2009}

%affine_lie_algebra
\section{Affine Lie algebras and their representations}
\label{sect:rep_aff_alg}
In this section, we recall the notion of affine Lie algebras and their representation theory.
Let $\mfrak{g}$ be a {{finite-dimensional}} Lie algebra that is simple or commutative
and $(\cdot|\cdot):\mfrak{g}\times\mfrak{g}\to\mbb{C}$ be a nondegenerate symmetric invariant bilinear form on $\mfrak{g}$.
The affinization $\what{\mfrak{g}}$ of $\mfrak{g}$ is defined by $\what{\mfrak{g}}=\mfrak{g}\otimes \mbb{C}[\zeta,\zeta^{-1}]\oplus\mbb{C}K$
with Lie brackets being defined by
\begin{equation}
	[X(m),Y(n)]=[X,Y](m+n)+m(X|Y)\delta_{m+n,0}K,\ \ [K,\what{\mfrak{g}}]=\{0\},
\end{equation}
where we denote $X\otimes \zeta^{n}$ by $X(n)$ for $X\in\mfrak{g}$ and $n\in\mbb{Z}$.
Let $M$ be a {{finite-dimensional}} representation of the {{finite-dimensional}} Lie algebra $\mfrak{g}$.
Then we lift the action of $\mfrak{g}$ to an action of a Lie subalgebra $\mfrak{g}\otimes \mbb{C}[\zeta]\oplus\mbb{C}K$ of the affine Lie algebra
so that $\mfrak{g}\otimes \zeta^{0}$ acts naturally, $\mfrak{g}\otimes \zeta\mbb{C}[\zeta]$ acts trivially, and $K$ acts as multiplication by a complex number $k$.
Then we obtain a representation $\what{M}_{k}$ of the affine Lie algebra $\what{\mfrak{g}}$ by
\begin{equation}
	\what{M}_{k}=\mrm{Ind}_{\mfrak{g}\otimes\mbb{C}[\zeta]\oplus\mbb{C}K}^{\what{\mfrak{g}}}M=U(\what{\mfrak{g}})\otimes_{U(\mfrak{g}\otimes\mbb{C}[\zeta]\oplus\mbb{C}K)}M.
\end{equation}
Here {{the}} introduced complex number $k$ is called the level of the representation.
By the Poincar\'{e}-Birkhoff-Witt theorem, $\what{M}_{k}$ is isomorphic to $U(\mfrak{g}\otimes \zeta^{-1}\mbb{C}[\zeta^{-1}])\otimes_{\mbb{C}}M$ as a vector space
or a $U(\mfrak{g}\otimes \zeta^{-1}\mbb{C}[\zeta^{-1}])$-module.

To classify {{finite-dimensional}} irreducible representations of $\mfrak{g}$, we assume that $\mfrak{g}$ is simple in this paragraph.
We fix a Cartan subalgebra $\mfrak{h}$ of $\mfrak{g}$,
and let $\Pi^{\vee}=\{\alpha_{i}^{\vee},\cdots,\alpha_{\ell}^{\vee}\}\subset\mfrak{h}$ be the set of simple coroots of $\mfrak{g}$.
Then the fundamental weights $\Lambda_{i}\in\mfrak{h}^{\ast}$ for $i=1,\cdots,\ell$ are defined by $\braket{\Lambda_{i},\alpha_{j}^{\vee}}=\delta_{ij}$,
and span the weight lattice $P=\bigoplus_{i=1}^{\ell}\mbb{Z}\Lambda_{i}$.
A weight $\Lambda\in P$ is called dominant if $\braket{\Lambda,\alpha_{i}^{\vee}}\ge \mbb{Z}_{\ge 0}$ for all $i=1,\cdots,\ell$.
We denote the set of dominant weights by $P_{+}$.
Finite-dimensional irreducible representations of $\mfrak{g}$ are labeled by $P_{+}${{. Namely}},
for a dominant weight $\Lambda\in P_{+}$, there is a {{finite-dimensional}} irreducible representation $L(\Lambda)$ of $\mfrak{g}$
with highest weight $\Lambda$, and conversely, the highest weight of a {{finite-dimensional}} irreducible representation of $\mfrak{g}$ is dominant.
For an irreducible representation $L(\Lambda)$ of $\mfrak{g}$, we can construct a representation $\what{L(\Lambda)}_{k}$ of $\what{\mfrak{g}}$
in the manner described in the previous paragraph.
Note that although $L(\Lambda)$ is irreducible as a representation of $\mfrak{g}$,
$\what{L(\Lambda)}_{k}$ is not necessarily an irreducible representation of $\what{\mfrak{g}}$,
then we denote by $L_{\mfrak{g}}(\Lambda,k)$ the irreducible quotient of $\what{L(\Lambda)}_{k}$ as a representation of $\what{\mfrak{g}}$.

In {{the case when}} $\mfrak{g}$ is commutative, the representation theory is {{simpler}}:
an irreducible representation $L(\Lambda)$ of $\mfrak{g}$ is one dimensional and labeled by an element $\Lambda\in \mfrak{g}^{\ast}$ so that
an element $X\in\mfrak{g}$ acts as $\braket{\Lambda,X}$ times the identity operator.
The corresponding representation $\what{L(\Lambda)}_{k}$ of $\what{\mfrak{g}}$, which we denote by $L_{\mfrak{g}}(\Lambda,k)$ is a Fock representation and irreducible.
Notice that Fock representations $L_{\mfrak{g}}(\Lambda,k)$ are all isomorphic if $k\neq 0$, thus we think that $k=1$ in $L_{\mfrak{g}}(\Lambda,k)$
if the finite-dimensional Lie algebra $\mfrak{g}$ is commutative.

On a representation space $L_{\mfrak{g}}(\Lambda,k)$ of an affine Lie algebra $\what{\mfrak{g}}$ constructed above,
we can define an action of the Virasoro algebra through the Segal-Sugawara construction.
We normalize the bilinear form so that $(\theta|\theta)=2$ if $\mfrak{g}$ is simple, where $\theta$ is the highest root of $\mfrak{g}$.
We define a number $h^{\vee}_{\mfrak{g}}$ by the dual Coxter number $h^{\vee}$ of $\mfrak{g}$ if $\mfrak{g}$ is simple,
and by $0$ if $\mfrak{g}$ is commutative, 
and assume that $k\neq -h^{\vee}_{\mfrak{g}}$.
Let $\{X_{a}\}_{a=1}^{\dim\mfrak{g}}$ be an orthonormal basis of $\mfrak{g}$ with respect to $(\cdot|\cdot)$.
Then the operators $L_{n}$ for $n\in\mbb{Z}$ acting on $L_{\mfrak{g}}(\Lambda,k)$ that are defined by
\begin{equation}
	\label{eq:segal_sugawara}
	L_{n}=\frac{1}{2(k+h^{\vee}_{\mfrak{g}})}\sum_{a=1}^{\dim\mfrak{g}}\sum_{k\in\mbb{Z}}\no{X_{a}(n-k)X_{a}(k)}
\end{equation}
give an action of the Virasoro algebra of central charge $c_{\mfrak{g},k}=\frac{k\dim\mfrak{g}}{k+h^{\vee}_{\mfrak{g}}}$.
Here the normal ordered product $\no{A(p)B(q)}$ is defined by $A(p)B(q)$ for $p<q$ and $B(q)A(p)$ for $p\ge q$.
Moreover a vector $v_{\Lambda}\in L(\Lambda)\hookrightarrow L_{\mfrak{g},k}(\Lambda)$ is an eigenvector of $L_{0}$ corresponding to 
an eigenvalue $h_{\Lambda}=\frac{(\Lambda|\Lambda+2\rho_{\mfrak{g}})}{2(k+h_{\mfrak{g}}^{\vee})}$, with $\rho_{\mfrak{g}}=\sum_{i=1}^{\ell}\Lambda_{i}$ if $\mfrak{g}$ is simple and
$\rho_{\mfrak{g}}=0$ if $\mfrak{g}$ is commutative.
{{The operator}} $L_{0}$ is diagonalizable on $L_{\mfrak{g}}(\Lambda,k)$ so that $L_{\mfrak{g}}(\Lambda,k)=\bigoplus_{n\in\mbb{Z}_{\ge 0}}L_{\mfrak{g}}(\Lambda,k)_{h_{\Lambda}+n}$
with each $L_{\mfrak{g}}(\Lambda,k)_{h}$ being the eigenspace of $L_{0}$ corresponding to an eigenvalue $h$.
{{This}} action of the Virasoro algebra is compatible with the action of $\what{\mfrak{g}}$
in the sense that $[L_{n},A\otimes f(\zeta)]=-A\otimes \zeta^{n+1}\frac{df(\zeta)}{d\zeta}$.

Among representations $L_{\mfrak{g}}(\Lambda,k)$, we can equip $L_{\mfrak{g}}(0,k)$ with a VOA structure.
The vacuum vector is $\ket{0}=1\otimes \vc{1}$, where $\vc{1}$ spans a one-dimensional representation $L(0)$ of $\mfrak{g}$.
Let $\{X_{a}\}_{a=1}^{\dim\mfrak{g}}$ be a basis of $\mfrak{g}$, then this VOA is strongly generated by vectors $X_{a}(-1)\ket{0}$.
In the following, we call this VOA the affine VOA of $\mfrak{g}$ with level $k$ and denote it by $L_{\mfrak{g},k}$.
Simple modules over $L_{\mfrak{g},k}$ are realized as highest weight representations $L_{\mfrak{g}}(\Lambda,k)$ of the same level.
For an $L_{\mfrak{g},k}$-module $M$, the invariance in Eq.(\ref{eq:invariant_form}) of a nondegenerate bilinear form $\braket{\cdot|\cdot}:M\times M\to\mbb{C}$ is rephrased as
$\braket{X(n)u|v}=-\braket{u|X(-n)v}$ for $u,v \in M$, $X\in\mfrak{g}$ and $n\in\mbb{Z}$.
Such an invariant bilinear form is specified on an irreducible representation $L_{\mfrak{g}}(\Lambda,k)$ by the normalization $\braket{v_{\Lambda}|v_{\Lambda}}=1$
with $v_{\Lambda}$ being the highest weight vector.

%internal_symmetry
\section{Internal symmetry}
\label{sect:internal_symmetry}
We again assume that $\mfrak{g}$ is a {{finite-dimensional}} complex Lie algebra that is simple or commutative.
Let $G$ be a {{finite-dimensional}} complex Lie group of which {{the}} Lie algebra is $\mfrak{g}$,
{\it i.e.}, it is a simple Lie group if $\mfrak{g}$ is simple and a torus if $\mfrak{g}$ is commutative.
To construct a generalization of SLE associated with a representation of an affine Lie algebra $\what{\mfrak{g}}$,
we consider the positive loop group $G(\mcal{O})=G[[\zeta^{-1}]]$ of $G$ as a group of internal symmetry.
A significant subgroup $G_{+}(\mcal{O})$ consists of elements that are the unit element modulo $G[[\zeta^{-1}]]\zeta^{-1}$.
The Lie algebras of $G(\mcal{O})$ and $G_{+}(\mcal{O})$ are $\mfrak{g}[[\zeta^{-1}]]$ and $\mfrak{g}[[\zeta^{-1}]]\zeta^{-1}$, respectively.
The group of automorphisms $\mrm{Aut}\mcal{O}$ acts on $G(\mcal{O})$ to define a semi-direct product $\mrm{Aut}\mcal{O}\ltimes G(\mcal{O})$.
Moreover, the subgroup $\mrm{Aut}_{+}\mcal{O}$ normalizes $G_{+}(\mcal{O})$, thus their semi-direct product $\mrm{Aut}_{+}\mcal{O}\ltimes G_{+}(\mcal{O})$ is also defined.

On a representation $L_{\frak{g}}(\Lambda,k)$ of the affine Lie algebra $\what{\mfrak{g}}$,
the Lie algebra $\mfrak{g}\otimes \mbb{C}[[\zeta^{-1}]]$ cannot act,
but its formal completion $\overline{L_{\mfrak{g}}(\Lambda,k)}=\prod_{n\in\mbb{Z}_{\ge 0}}L_{\mfrak{g}}(\Lambda,k)_{h_{\Lambda}+n}$
admits an action of $\mfrak{g}\otimes \mbb{C}[[\zeta^{-1}]]$.
It is also obvious that the action of $\mfrak{g}\otimes \mbb{C}[[\zeta^{-1}]]$ is exponentiated to define an action of $G(\mcal{O})$.
Indeed, an element in $\mfrak{g}\otimes \zeta^{-1}\mbb{C}[[\zeta^{-1}]]$ strictly raises {{the}} degree,
and a zero-mode element $X\otimes \zeta^{0}$ is exponentiated to be an action of $e^{X}\in G$
while each homogeneous space is a representation of the {{finite-dimensional}} Lie group $G$.
Moreover, this action of $G(\mcal{O})$ is compatible with the action of $\mrm{Aut}\mcal{O}$ due to the Segal-Sugawara construction.
Thus $\mrm{Aut}_{+}\mcal{O}\ltimes G_{+}(\mcal{O})$ acts on $\overline{L_{\mfrak{g}}(\Lambda,k)}$.

We investigate how each field is transformed under the adjoint action of $e^{\bf a}$ where ${\bf a}=A\otimes a(\zeta)\in\mfrak{g}\otimes \mbb{C}[[\zeta^{-1}]]$.
We compute the commutator $[{\bf a},Y(B,w)]$ for $B\in L_{\mfrak{g},k}$.
From the {{operator product expansion (OPE)}} formula
\begin{equation}
	[Y(A(-1)\ket{0},z),Y(B,w)]=\sum_{k\ge 0}Y(A(k)B,w)\del_{w}^{(k)}\delta(z-w),
\end{equation}
we obtain
\begin{equation}
	[A(n),Y(B,w)]=\sum_{k\ge 0}\binom{n}{k}w^{n-k}Y(A(k)B,w).
\end{equation}
Thus the desired commutator is computed as
\begin{equation}
	[{\bf a},Y(B,w)]=Y({\bf a}_{w}B,w),
\end{equation}
where ${\bf a}_{w}=\sum_{k\ge 0}\del^{(k)}a(w)A(k)$.
This enables us to obtain the following transformation formula{{:}}
\begin{equation}
	Y(B,w)=e^{\bf a}Y(e^{-{\bf a}_{w}}B,w)e^{-{\bf a}}
\end{equation}

Now we compute $e^{-{\bf a}_{w}}X(-1)\ket{0}$ for some $X\in \mfrak{g}$
to investigate the transformation rule of $Y(X(-1)\ket{0},z)$ under the adjoint action by $e^{{\bf a}}$.
The action of ${\bf a}_{w}$ on $X(-1)\ket{0}$ gives
\begin{align}
	{\bf a}_{w}X(-1)\ket{0}
		&=a(w)(\mrm{ad}A)(X)(-1)\ket{0}+k(A|X)\del a(w)\ket{0}.
\end{align}
Applying ${\bf a}_{w}$ once more{{:}}
\begin{align}
	{\bf a}_{w}^{2}X(-1)\ket{0}
		&=a(w)^{2}(\mrm{ad}A)^{2}(X)(-1)\ket{0},
\end{align}
where we have used the invariance of the bilinear form $(A|[A,X])=([A,A]|X)=0$,
and inductively,
\begin{equation}
	{\bf a}_{w}^{n}X(-1)\ket{0}=a(w)^{n}(\mrm{ad}A)^{n}(X)(-1)\ket{0}
\end{equation}
for $n\ge 2$.
Thus we can see that
\begin{equation}
	e^{-{\bf a}_{w}}X(-1)\ket{0}
	=(e^{-a(w)\mrm{ad}A}X)(-1)\ket{0}-k(A|X)\del a(w)\ket{0},
\end{equation}
which implies that
\begin{equation}
	\label{eq:internal_twist_field}
	Y(X(-1)\ket{0},w)=e^{\bf a}Y\left((e^{-a(w)\mrm{ad}A}X)(-1)\ket{0},w\right)e^{-{\bf a}}-k(A|X)\del a(w).
\end{equation}
It is also convenient to {{note}} the formula for the object {{in the form}} $e^{-{\bf a}}X\otimes x(\zeta)e^{{\bf a}}$,
where ${\bf a}=A\otimes a(\zeta)\in\mfrak{g}\otimes\mbb{C}[[\zeta^{-1}]]$, $x(\zeta)\in\mbb{C}((\zeta^{-1}))$ and $X\in\mfrak{g}$ are taken as above.
{{This}} becomes
\begin{align}
	\label{eq:internal_twist_algebra}
	e^{-{\bf a}}X\otimes x(\zeta)e^{\bf a}
		&=\mrm{Res}_{w}\sum_{n\in\mbb{Z}}(e^{-a(w)\mrm{ad}A}X)\otimes \zeta^{n}w^{-n-1}x(w) 
		-k(A|X)\mrm{Res}_{w}\del a(w)x(w) \notag \\
		&=\sum_{m=0}^{\infty}\frac{(-1)^{m}}{m!}(\mrm{ad}A)^{m}(X)\otimes a(\zeta)^{m}x(\zeta)
		-k(A|X)\mrm{Res}_{w}\del a(w)x(w).
\end{align}

We next investigate the transformation rule of the Virasoro field $L(z)$ under the action of $G(\mcal{O})$.
{{We}} compute $e^{-{\bf a}_{z}}L_{-2}\ket{0}$ where ${\bf a}=A\otimes a(\zeta)\in \mfrak{g}\otimes\mbb{C}[[\zeta^{-1}]]$
and correspondingly ${\bf a}_{z}=\sum_{k\ge 0}\del^{(k)}a(z)A(k)$.
{{Note}} that the OPE
\begin{equation}
	[L(z),Y(A(-1)\ket{0},w)]=Y(A(-1)\ket{0},w)\del_{w}\delta(z-w)+\del Y(A(-1)\ket{0},w)\delta(z-w)
\end{equation}
is equivalent to
\begin{equation}
	[Y(A(-1)\ket{0},z),L(w)]=Y(A(-1)\ket{0},w)\del_{w}\delta(z-w),
\end{equation}
which implies
\begin{equation}
	A(n)L_{-2}\ket{0}=
	\begin{cases}
		A(-1)\ket{0}, 	& n=1,\\
		0, & n\in\mbb{Z}_{\ge 0}\backslash \{1\}.
	\end{cases}
\end{equation}
Thus
\begin{equation}
	-{\bf a}_{z}L_{-2}\ket{0}=-\del a(z) A(-1)\ket{0}.
\end{equation}
If we apply $-{\bf a}_{z}$ {{one}} more time,
\begin{equation}
	(-{\bf a}_{z})^{2}L_{-2}\ket{0}=k(\del a(z))^{2}(A|A)\ket{0}.
\end{equation}
Then we obtain the following transformation formula:
\begin{equation}
	\label{eq:virasoro_internal}
	L(z)=e^{\bf a}L(z)e^{-{\bf a}}-\del a(z)e^{\bf a}Y(A(-1)\ket{0},z)e^{-{\bf a}}+\frac{k(A|A)(\del a(z))^{2}}{2}.
\end{equation}

\subsection{{{Formulas}} in {{the}} case of commutative $\mfrak{g}$}

Let us {{note the formilas}} in Eq.(\ref{eq:internal_twist_field}) and Eq.(\ref{eq:internal_twist_algebra}) in a more explicit way in {{the case when}} $\mfrak{g}$ is commutative.
{{Now,}} we have $[A,X]=0$ for any $A, X\in \mfrak{g}$, which implies that
\begin{align}
	\label{eq:current_internal_commutative}
	X(z)&=e^{\bf a}X(z)e^{-{\bf a}}-k(A|X)\del a(z),\\
	e^{-{\bf a}}X\otimes x(\zeta)e^{\bf a}&=X\otimes x(\zeta)-k(A|X)\mrm{Res}_{w}\del a(w)x(w).
\end{align}

\subsection{{{Formulas}} in $\mfrak{g}=\mfrak{sl}_{2}$ case}
\label{subsect:formulae_sl2}
We now {{focus}} our attention on the case of $\mfrak{g}=\mfrak{sl}_{2}$ and explicitly {{note the formulas in}}
Eq.(\ref{eq:internal_twist_field}) and Eq.(\ref{eq:internal_twist_algebra}).
We take as a standard basis of $\mfrak{sl}_{2}$
\begin{align}
	E&=\left(\begin{array}{cc}
			0 & 1\\
			0 & 0
		\end{array}\right),&
	H&=\left(\begin{array}{cc}
			1 & 0\\
			0 & -1
		\end{array}\right), &
	F&=\left(\begin{array}{cc}
			0 & 0\\
			1 & 0
		\end{array}\right),
\end{align}
and denote $E\otimes e(\zeta)$, $H\otimes h(\zeta)$ and $F\otimes f(\zeta)$ for $e(\zeta), h(\zeta), f(\zeta)\in \mbb{C}[[\zeta^{-1}]]$ simply by
${\bf e}$, ${\bf h}$ and ${\bf f}$, respectively.
We also write a current field $Y(X(-1)\ket{0},z)$ by $X(z)$ for $X\in\mfrak{g}$.

\begin{enumerate}
\item 	$X=A=H$.
		\begin{align*}
			H(z)&=e^{\bf h}H(z)e^{-{\bf h}}-2k\del h(z), \\
			e^{-{\bf h}}H\otimes x(\zeta)e^{\bf h}&=H\otimes x(\zeta)-2k\mrm{Res}_{w}\del h(w)x(w).
		\end{align*}
\item 	$X=H$, $A=E$.
		\begin{align*}
			H(z)&=e^{\bf e}H(z)e^{-{\bf e}}+2e(z)E(z),\\
			e^{-{\bf e}}H\otimes x(\zeta)e^{\bf e}&=H\otimes x(\zeta)+2E\otimes e(\zeta)x(\zeta).
		\end{align*}
\item 	$X=H$, $A=F$.
		\begin{align*}
			H(z)&=e^{\bf f}H(z)e^{-{\bf f}}-2f(z)F(z),\\
			e^{-{\bf f}}H\otimes x(\zeta)e^{\bf f}&=H\otimes x(\zeta)-2F\otimes f(\zeta)x(\zeta).
		\end{align*}
\item 	$X=E$, $A=H$.
		\begin{align*}
			E(z)&=e^{-2h(z)}e^{\bf h}E(z)e^{-{\bf h}}, \\
			e^{-{\bf h}}E\otimes x(\zeta)e^{\bf h}&=E\otimes e^{-2h(\zeta)}x(\zeta).
		\end{align*}
\item 	$X=A=E$.
		\begin{align*}
			E(z)&=e^{\bf e}E(z)e^{-{\bf e}}, \\
			e^{-{\bf e}}E\otimes x(\zeta)e^{\bf e}&=E\otimes x(\zeta).
		\end{align*}
\item 	$X=E$, $A=F$.
		\begin{align*}
			E(z)&=e^{\bf f}E(z)e^{-{\bf f}}+f(z)e^{\bf f}H(z)e^{-{\bf f}}-f(z)^{2}e^{\bf f}F(z)e^{-{\bf f}}-k\del f(z),\\
			e^{-{\bf f}}E\otimes x(\zeta)e^{\bf f}&=E\otimes x(\zeta)+H\otimes f(\zeta)x(\zeta)-F\otimes f(\zeta)^{2}x(\zeta)-k\mrm{Res}_{w}\del f(w)x(w).
		\end{align*}
\item 	$X=F$, $A=H$.
		\begin{align*}
			F(z)&=e^{2h(z)}e^{\bf h}F(z)e^{-{\bf h}}, \\
			e^{-{\bf h}}F\otimes x(\zeta)e^{\bf h}&=F\otimes e^{2h(\zeta)}x(\zeta).
		\end{align*}
\item 	$X=F$, $A=E$.
		\begin{align*}
			F(z)&=e^{\bf e}F(z)e^{-{\bf e}}-e(z)e^{\bf e}H(z)e^{-{\bf e}}-e(z)^{2}E(z)-k\del e(z),\\
			e^{-{\bf e}}F\otimes x(\zeta)e^{\bf e}&=F\otimes x(\zeta) -H \otimes e(\zeta)x(\zeta)-E\otimes e(\zeta)^{2}x(\zeta)-k\mrm{Res}_{w}\del e(w)x(w).
		\end{align*}
\item 	$X=A=F$.
		\begin{align*}
			F(z)&=e^{\bf f}F(z)e^{-{\bf f}}, \\
			e^{-{\bf f}}F\otimes x(\zeta)e^{\bf f}&=F\otimes x(\zeta).
		\end{align*}
\end{enumerate}

%random_process
\section{Construction of a random process}
\label{sect:random_process}
In this section, we construct a random process on the {{infinite-dimensional}} Lie group $\mrm{Aut}_{+}\mcal{O}\ltimes G_{+}(\mcal{O})$,
which {{was}} introduced in the previous Sect. \ref{sect:internal_symmetry}.
It is a natural generalization of the random process on $\mrm{Aut}_{+}\mcal{O}$,
which was the fundamental object in the group theoretical formulation of {{SLEs}} in Sect. \ref{sect:group_theoretical},
to a case with internal symmetry.

\subsection{General Lie algebras $\mfrak{g}$}
We shall construct a random process that is a generalization of {{SLEs}} with internal symmetry described by $G_{+}(\mcal{O})$.
Such a random process is expected to be induced from a random process on an {{infinite-dimensional}} Lie group $\mrm{Aut}_{+}\mcal{O}\ltimes G_{+}(\mcal{O})$.
To decide a direction {{for}} designing a random process on this group, we first make an observation on an annihilator of the vacuum vector
in the vacuum representation $L_{\mfrak{g},k}$.
Since we have defined a representation of the Virasoro algebra by the Segal-Sugawara construction,
$L_{-2}\ket{0}=\frac{1}{2(k+h_{\mfrak{g}}^{\vee})}\sum_{r=1}^{\dim\mfrak{g}}X_{r}(-1)^{2}\ket{0}$.
Combining the fact that the vacuum vector is translation invariant, we see that the operator
\begin{equation}
	-2L_{-2}+\frac{\kappa}{2}L_{-1}^{2}+\frac{1}{k+h_{\mfrak{g}}^{\vee}}\sum_{r=1}^{\dim\mfrak{g}}X_{r}(-1)^{2}
\end{equation}
annihilates the vacuum vector for arbitrary $\kappa$.
We now assume that the highest weight vector $v_{\Lambda}$ of a representation $L_{\mfrak{g}}(\Lambda,k)$ is annihilated by an operator of the form
\begin{equation}
	\label{eq:annihilator_affine}
	-2L_{-2}+\frac{\kappa}{2}L_{-1}^{2}+\frac{\tau}{2}\sum_{r=1}^{\dim\mfrak{g}}X_{r}(-1)^{2}
\end{equation}
with parameters $\kappa$ and $\tau$ being finely tuned positive numbers.
The existence of such an annihilator of the above form will be discussed later in Sect. \ref{sect:annihilator}.

We consider a random process $\scr{G}_{t}$ on $\mrm{Aut}_{+}\mcal{O}\ltimes G_{+}(\mcal{O})$ that satisfies the following SDE
\begin{equation}
	\label{eq:random_process_affine}
	\scr{G}_{t}^{-1}d\scr{G}_{t}=\left(-2L_{-2}+\frac{\kappa}{2}L_{-1}^{2}+\frac{\tau}{2}\sum_{r=1}^{\dim\mfrak{g}}X_{r}(-1)^{2}\right)dt
		+L_{-1}dB_{t}^{(0)}+\sum_{r=1}^{\dim\mfrak{g}}X_{r}(-1)dB_{t}^{(r)},
\end{equation}
where $B_{t}^{(i)}$ for $i=0,1,\cdots,\dim\mfrak{g}$ are mutually independent Brownian motions with variance $\kappa$ for $B_{t}^{(0)}$
and $\tau$ for $B_{t}^{(r)}$ with $r=1,\cdots,\dim\mfrak{g}$.
{{An idea for}} considering a random process on such an {{infinite-dimensional}} Lie group as $\mrm{Aut}_{+}\mcal{O}\ltimes G_{+}(\mcal{O})$
has already appeared in the work by {{Rasmussen}}.\cite{Rasmussen2007}
but it lacks an SDE based on an annihilating operator,
and it does not include the classical SLE in the coordinate transformation part.

\begin{prop}
\label{prop:martingale_affine}
Assume that the highest weight vector $v_{\Lambda}$ of $L_{\mfrak{g}}(\Lambda,k)$ is annihilated by the operator in Eq.(\ref{eq:annihilator_affine}).
Then for a random process $\scr{G}_{t}$ on $\mrm{Aut}_{+}\mcal{O}\ltimes G_{+}(\mcal{O})$ satisfying Eq.(\ref{eq:random_process_affine}),
the random process $\scr{G}_{t}v_{\Lambda}$ in $\overline{L_{\mfrak{g}}(\Lambda,k)}$ is a local martingale.
\end{prop}

We can write the random process $\scr{G}_{t}$ as $\scr{G}_{t}=\Theta_{t}Q(\rho_{t})$
where the random process $\rho_{t}$ on $\mrm{Aut}_{+}\mcal{O}$ induces the SLE$(\kappa)$
and $\Theta_{t}$ is a random process on $G_{+}(\mcal{O})$.

\begin{prop}
Under the ansatz $\scr{G}_{t}=\Theta_{t}Q(\rho_{t})$ described above,
the random process $\Theta_{t}$ on $G_{+}(\mcal{O})$ satisfies the SDE
\begin{equation}
	\label{eq:sde_internal_sym}
	\Theta_{t}^{-1}d\Theta_{t}=\frac{\tau}{2}\sum_{r=1}^{\dim\mfrak{g}}(X_{r}\otimes \rho_{t}(\zeta)^{-1})^{2}dt +\sum_{r=1}^{\dim\mfrak{g}}X_{r}\otimes \rho_{t}(\zeta)^{-1} dB_{t}^{(r)}.
\end{equation}
\end{prop}
\begin{proof}
The action of the Virasoro algebra on an affine Lie algebra, which is described by the relation $[L_{n},X(m)]=-mX(n+m)$,
implies the transformation formula
\begin{equation}
	G(\rho)X\otimes f(\zeta)G(\rho)^{-1}=X\otimes f(\rho(\zeta))
\end{equation}
for $f(\zeta)\in\mbb{C}((\zeta^{-1}))$ and $\rho\in\mrm{Aut}_{+}\mcal{O}$.
If we apply this formula in the case {{when}} $f(\zeta)=\zeta^{-1}$, we obtain the desired result.
\end{proof}

{{Equation}} (\ref{eq:sde_internal_sym}) has already appeared in an equivalent form in the correlation function formulation of SLEs corresponding to WZW models. \cite{BettelheimGruzbergLudwigWiegmann2005, AlekseevBytskoIzyurov2011}
Let $\mcal{Y}(-,z)$ be an intertwining operator of type $\binom{L_{\mfrak{g}}(\Lambda_{3},k)}{L_{\mfrak{g}}(\Lambda_{1},k)\ \ L_{\mfrak{g}}(\Lambda_{2},k)}$,
and $v\in L(\Lambda_{1})$ be a primary vector in the top space of $L_{\mfrak{g}}(\Lambda_{1},k)$.
If we take {{the}} adjoint of the primary field $\mcal{Y}(v,z)$ by $\scr{G}_{t}^{-1}$, we obtain
\begin{equation}
	\label{eq:random_transformation_affine}
	\scr{G}_{t}^{-1}\mcal{Y}(v,z)\scr{G}_{t}=\mcal{Y}(\Theta_{t}^{-1}(z)v,\rho_{t}(z))(\del \rho_{t}(z))^{h_{\Lambda_{1}}}.
\end{equation}
Here the object $\Theta_{t}^{-1}(z)$ is a random process on the group of $z^{-1}\mbb{C}[[z^{-1}]]$-points in $G$
obtained by substituting $\zeta=z$ in $\Theta_{t}^{-1}$.
From the identity $\Theta_{t}^{-1}\Theta_{t}=\mrm{Id}$, the SDE on $\Theta_{t}^{-1}(z)$ becomes
\begin{equation}
	d\Theta_{t}^{-1}(z)\Theta_{t}(z)=\frac{\tau}{2}\sum_{r=1}^{\dim\mfrak{g}}(\rho_{t}(z)^{-1}X_{r})^{2}-\sum_{r=1}^{\dim\mfrak{g}}\rho_{t}(z)^{-1}X_{r}dB_{t}^{(r)}.
\end{equation}
Apart from the Jacobian part, the {{right-hand}} side of Eq.(\ref{eq:random_transformation_affine}) is
the random transformation of a primary field in Eq.(\ref{eq:sle_wzw_correlation_function}) considered in the correlation function formulation of SLEs,
\cite{BettelheimGruzbergLudwigWiegmann2005, AlekseevBytskoIzyurov2011} which seemed to be {\it ad hoc},
while it naturally appears in the group theoretical formulation presented here.

The SDE in Eq.(\ref{eq:sde_internal_sym}) {{for}} the random process along {{the}} internal symmetry
is still not {{sufficient}} to compute matrix elements like $\braket{u|\scr{G}_{t}|v_{\Lambda}}$.
In the following two subsections, we construct the random process $\Theta_{t}$ in the most explicit way
in {{the cases when}} $\mfrak{g}$ is commutative and $\mfrak{g}=\mfrak{sl}_{2}$.

\subsection{{{Case when}} $\mfrak{g}$ is commutative}
We temporary denote the dimension of $\mfrak{g}$ by $\ell$.
Let $H_{1},\cdots,H_{\ell}$ be an orthonormal basis of $\mfrak{g}$ with respect to the bilinear form $(\cdot|\cdot)$.
We put an ansatz on $\Theta_{t}$ as
\begin{equation}
	\Theta_{t}=e^{H_{1}\otimes h^{1}_{t}(\zeta)}\cdots e^{H_{\ell}\otimes h^{\ell}_{t}(\zeta)},
\end{equation}
where $h^{i}_{t}(\zeta)$ are $\mbb{C}[[\zeta^{-1}]]\zeta^{-1}$-valued random {{processes}}.
\begin{prop}
Under the above ansatz on $\Theta_{t}$, the random processes $h^{i}_{t}(\zeta)$ satisfy
\begin{equation}
	\label{eq:sde_internal_Heisenbeg}
	dh^{i}_{t}(\zeta)=\frac{1}{\rho_{t}(\zeta)}dB_{t}^{(i)}
\end{equation}
for $i=1,\cdots,\ell$.
\end{prop}
Thus the random processes $h^{i}_{t}(\zeta)$ are completely determined by the solution of {{SLEs}} so that $h^{i}_{t}(\zeta)=\int_{0}^{t}\frac{dB^{(i)}_{s}}{\rho_{s}(\zeta)}$.

\subsection{Specialization to $\mfrak{sl}_{2}$}
To construct the random process $\Theta_{t}$ in a sufficiently explicit way,
we make an ansatz that it is written as $\Theta_{t}=e^{{\bf e}_{t}}e^{{\bf h}_{t}}e^{{\bf f}_{t}}$,
where ${\bf e}_{t}=E\otimes e_{t}(\zeta)$, ${\bf h}_{t}=H\otimes h_{t}(\zeta)$, {{and}} ${\bf f}_{t}=F\otimes f_{t}(\zeta)$ are random processes on $\mfrak{g}\otimes\mbb{C}[[\zeta^{-1}]]\zeta^{-1}$
associated with $\mbb{C}[[\zeta^{-1}]]\zeta^{-1}$-valued random processes $e_{t}(\zeta)$, $h_{t}(\zeta)$, and $f_{t}(\zeta)$.
Then we shall derive SDEs on $e_{t}(\zeta)$, $h_{t}(\zeta)$ and $f_{t}(\zeta)$.
{{We}} assume the SDEs as
\begin{equation}
	dx_{t}(\zeta)=\overline{x}_{t}(\zeta)dt+\sum_{r=1}^{3}x_{t}^{r}(\zeta)dB_{t}^{(r)},\ \ x=e,h,f.
\end{equation}
Since $X(n)$ with $n<0$ are mutually commutative for a fixed $X\in\mfrak{sl}_{2}$,
the increment of the random process $\Theta_{t}$ is computed by the standard Ito calculus,
and we can determine data $\overline{x}_{t}(\zeta)$ and $x_{t}^{(r)}(\zeta)$ so the increment of $\Theta_{t}$ {{is in}} the desired form in Eq.(\ref{eq:sde_internal_sym}).
After {{the}} computation that is presented in Appendix \ref{sect:app_SDE}, we obtain the following {{proposition}}
\begin{prop}
\label{prop:sde_internal_sl2}
Under the ansatz  $\Theta_{t}=e^{{\bf e}_{t}}e^{{\bf h}_{t}}e^{{\bf f}_{t}}$ described above,
the SDE in Eq.(\ref{eq:sde_internal_sym}) implies the following set of SDEs:
\begin{align}
	de_{t}(\zeta)=&-\frac{e^{2h_{t}(\zeta)}}{\sqrt{2}\rho_{t}(\zeta)}dB_{t}^{(2)}-\frac{ie^{2h_{t}(\zeta)}}{\sqrt{2}\rho_{t}(\zeta)}dB_{t}^{(3)}, \\
	dh_{t}(\zeta)=&-\frac{\tau}{2\rho_{t}(\zeta)^{2}}dt
				-\frac{1}{\sqrt{2}\rho_{t}(\zeta)}dB_{t}^{(1)}+\frac{f_{t}(\zeta)}{\sqrt{2}\rho_{t}(\zeta)}dB_{t}^{(2)}+\frac{if_{t}(\zeta)}{\sqrt{2}\rho_{t}(\zeta)}dB_{t}^{(3)}, \\
	df_{t}(\zeta)=&-\frac{\sqrt{2}f_{t}(\zeta)}{\rho_{t}(\zeta)}dB_{t}^{(1)}-\frac{1-f_{t}(\zeta)^{2}}{\sqrt{2}\rho_{t}(\zeta)}dB_{t}^{(2)}+\frac{i(1+f_{t}(\zeta)^{2})}{\sqrt{2}\rho_{t}(\zeta)}dB_{t}^{(3)}.		
\end{align}
\end{prop}

%annihilating_operator
\section{Annihilating operator of a highest weight vector}
\label{sect:annihilator}
We have assumed in Sect.\ref{sect:random_process} that the highest weight vector $v_{\Lambda}$ of $L_{\mfrak{g}}(\Lambda,k)$
is annihilated by an operator of the form in Eq.(\ref{eq:annihilator_affine}) with finely tuned parameters $\kappa$ and $\tau$.
In this section we see examples of such annihilating operators.
As we have already seen, the vacuum vector $\ket{0}$ is annihilated by the operator in Eq.(\ref{eq:annihilator_affine})
for $\tau=\frac{2}{k+h_{\mfrak{g}}^{\vee}}$ and arbitrary $\kappa$.
Thus we shall search for an example acting on a ``charged" representation.

\subsection{{{Case when}} $\mfrak{g}$ is commutative}
We first compute vectors $L_{-2}v_{\Lambda}$ and $L_{-1}^{2}v_{\Lambda}$.
By the expression of $L_{n}$ via the Segal-Sugawara construction in Eq.(\ref{eq:segal_sugawara}), they {{can be}} computed as
\begin{align}
	L_{-2}v_{\Lambda}&=\left(\frac{1}{2}\sum_{i=1}^{\ell}H_{i}(-1)^{2}+\Lambda(-2)\right)v_{\Lambda}, \\
	L_{-1}^{2}v_{\Lambda}&=\left(\Lambda(-1)^{2}+\Lambda(-2)\right)v_{\Lambda}.
\end{align}
Here we have identified $\mfrak{g}^{\ast}$ with $\mfrak{g}$ via the nondegenerate bilinear form $(\cdot|\cdot)$.
We assume that $\Lambda$ is proportional {{to}} $H_{1}$ with {{the}} coefficient being written as $\lambda$: $\Lambda=\lambda H_{1}$.
Under this assumption,
\begin{equation}
	\left(-2L_{-2}+\frac{\kappa}{2}L_{-1}^{2}\right)v_{\lambda}=\left(-(1-2\lambda^{2})H_{1}(-1)^{2}-\sum_{i=2}^{\ell}H_{i}(-1)^{2}\right)v_{\Lambda},
\end{equation}
for $\kappa=4$.
Thus we have found an operator that annihilates $v_{\Lambda}$ of a suitable form.
\begin{prop}
\label{prop:annihilator_Heisenberg}
The following operator annihilates the highest weight vector $v_{\Lambda}$ for $\Lambda=\lambda H_{1}$:
\begin{equation}
	\label{eq:annihilator_Heisenberg}
	-2L_{-2}+\frac{\kappa}{2}L_{-1}^{2}+\frac{1}{2}\sum_{i=1}^{\ell}\tau_{i}H_{i}(-1)^{2},
\end{equation}
where $\kappa=4$, $\tau_{1}=2-4\lambda^{2}$ and $\tau_{i}=2$ for $i\ge 2$.
\end{prop}

\subsection{{{Case when}} $\mfrak{g}=\mfrak{sl}_{2}$}
Here we assume that the level is $k=1$.
In this case, the vacuum representation $L_{\mfrak{sl}_{2},1}$ is isomorphic as a VOA to the lattice {{VOA}} $V_{Q}$
associated with the root lattice $Q=\mbb{Z}\alpha$, $(\alpha|\alpha)=2$ of $\mfrak{sl}_{2}$.
The isomorphism is described by
\begin{align}
	E(z)&\mapsto \Gamma_{\alpha}(z),&
	H(z)&\mapsto \alpha(z),&
	F(z)&\mapsto \Gamma_{-\alpha}(z).
\end{align}
Here $\alpha(z)$ is the free Bose field and $\Gamma_{\pm\alpha}(z)$ are the vertex operators associated with $\pm\alpha\in Q$.
This isomorphism of VOAs is called the Frenkel-Kac construction of an affine VOA,\cite{FrenkelKac1980}
of which an exposition is also contained in Appendix \ref{subsect:app_frenkel_kac}.
The dominant weights of level $k=1$ are exhausted by $0$ and the fundamental weight $\Lambda$ such that $(\Lambda|\alpha)=1$.
The spin-$\frac{1}{2}$ representation $L_{\mfrak{g}}(\Lambda,1)$ corresponding to $\Lambda$ is also realized as a module of the lattice VOA $V_{Q}$ by $V_{Q+\Lambda}${{, which}} is defined by
\begin{equation}
	V_{Q+\Lambda}= \bigoplus_{\beta\in Q}L_{\mbb{C}\otimes_{\mbb{Z}}Q}(0,1)\otimes e^{\beta+\Lambda}.
\end{equation}
Here $L_{\mbb{C}\otimes_{\mbb{Z}}Q}(0,1)$ is the vacuum Fock space introduced in Sect.\ref{sect:rep_aff_alg}.
Let the top space of $L_{\mfrak{sl}_{2}}(\Lambda,1)$ be realized as  $L(\Lambda)=\mbb{C}v_{0}\oplus\mbb{C}v_{1}$ so that $Ev_{0}=0$.
Then the isomorphism $L_{\mfrak{sl}_{2}}(\Lambda,1)\simeq V_{Q+\Lambda}$ is determined by
\begin{equation}
	v_{0}\mapsto e^{\Lambda},\ \ v_{1}\mapsto e^{-\Lambda}.
\end{equation}

We show that both $v_{0}$ and $v_{1}$ {{are}} annihilated by an operator of the form in Eq.(\ref{eq:annihilator_affine}).
Let $\mcal{Y}(-,z)$ be the intertwining operator of type $\binom{L_{\mfrak{sl}_{2}}(\Lambda,1)}{L_{\mfrak{sl}_{2}}(\Lambda,1)\ \ L_{\mfrak{sl}_{2},1}}$.
Then we have $\mcal{Y}(e^{\pm\Lambda},z)=\Gamma_{\pm \Lambda}(z)$, where $\Gamma_{\pm \Lambda}(z)$ are generalized vertex operators associated with $\pm\Lambda$.
Such a realization of an intertwining operator allows us to obtain
\begin{align}
	L_{-2}e^{\pm\Lambda}=L_{-1}^{2}e^{\pm\Lambda}=\left(\frac{1}{4}\alpha(-1)^{2}\pm \frac{1}{2}\alpha(-2)\right)e^{\pm\Lambda}
\end{align}
by computation of operator product expansions.
In the case of $\mfrak{g}=\mfrak{sl}_{2}$, then
\begin{equation}
	\sum_{r=1}^{3}X_{r}(-1)^{2}=\frac{1}{2}H(-1)^{2}+E(-1)F(-1)+F(-1)E(-1).
\end{equation}
It is obvious that $E(-1)e^{\Lambda}=0$ from $\Gamma_{\alpha}(z)\Gamma_{\Lambda}(w)=(z-w)\Gamma_{\alpha,\Lambda}(z,w)$.
On the other hand, $F(-1)$ nontrivially acts on $e^{\Lambda}$ and further applying $E(-1)$, we have $E(-1)F(-1)e^{\Lambda}=\alpha(-2)e^{\Lambda}$.
Combining them we can see that
\begin{equation}
	\left(-2L_{-2}+\frac{\kappa}{2}L_{-1}^{2}+\frac{\tau}{2}\sum_{r=1}^{3}X_{r}(-1)^{2}\right)e^{\Lambda}=0
\end{equation}
if the relation $\kappa+2\tau -4=0$ holds.
Computation for $e^{-\Lambda}$ is carried {{out}} in an analogous way.
We have $F(-1)e^{-\Lambda}=0$, while $F(-1)E(-1)e^{-\Lambda}=-\alpha(-2)e^{-\Lambda}$, which leads us to
\begin{equation}
	\left(-2L_{-2}+\frac{\kappa}{2}L_{-1}^{2}+\frac{\tau}{2}\sum_{r=1}^{3}X_{r}(-1)^{2}\right)e^{-\Lambda}=0
\end{equation}
if the parameters $\kappa$ and $\tau$ {{satisfy}} the same relation $\kappa+2\tau-4=0$ as in the case of $e^{\Lambda}$.

We summarize the above computation as follows.
\begin{prop}
\label{prop:annihilator_sl2}
Let $\Lambda$ be the fundamental weight of $\mfrak{sl}_{2}$,
and {{let}} the fundamental representation of $\mfrak{sl}_{2}$ be described by $L(\Lambda)=\mbb{C}v_{\Lambda}\oplus \mbb{C}Fv_{\Lambda}$.
Here $v_{\Lambda}$ is the highest weight vector of highest weight $\Lambda$.
We also denote the vector $Fv_{\Lambda}$ by $v_{-\Lambda}$.
Then we have in $L_{\mfrak{sl}_{2}}(\Lambda,1)$
\begin{equation}
	\left(-2L_{-2}+\frac{\kappa}{2}L_{-1}^{2}+\frac{\tau}{2}\sum_{r=1}^{3}X_{r}(-1)^{2}\right)v_{\pm\Lambda}=0
\end{equation}
if the relation $\kappa+2\tau-4=0$ holds.
\end{prop}

%local_martingale
\section{Local martingales}
As an application of construction of a random process $\scr{G}_{t}$ on an {{infinite-dimensional}} Lie group presented in Sect. \ref{sect:random_process},
we compute several local martingales associated with the solution of {{SLEs}} with internal degrees of freedom
by taking the inner product $\braket{u|\scr{G}_{t}|v_{\Lambda}}$.

\label{sect:martingale}
\subsection{{{Case when}} $\mfrak{g}$ is commutative}
The local martingale $\scr{G}_{t}v_{\Lambda}$ on $\overline{L_{\mfrak{g}}(\Lambda,1)}$ generates local martingales
when we take the inner product of it with any vectors in $L_{\mfrak{g}}(\Lambda,1)$.
To describe them explicitly, we first investigate how a current field $H(z)$ and the Virasoro field $L(z)$ are transformed under the adjoint action by $\scr{G}_{t}$.
First a current field $H(z)$ transforms under the adjoint action by $e^{-{\bf h}^{1}_{t}}$ as in Eq.(\ref{eq:current_internal_commutative}),
which implies
\begin{equation}
	\Theta_{t}^{-1}H(z)\Theta_{t}=H(z)-\sum_{i=1}^{\ell}(H_{i}|H)\del h^{i}_{t}(z).
\end{equation}
Since the transformation rule of $H(z)$ under the adjoint action by $Q(\rho_{t})^{-1}$ has been already obtained,
\begin{equation}
	\scr{G}_{t}^{-1}H(z)\scr{G}_{t}=H(\rho_{t}(z))\pr{\rho}_{t}(z)-\sum_{i=1}^{\ell}(H_{i}|H)\del h^{i}_{t}(z).
\end{equation}
This can be used to {{formulate}} a local martingale $\braket{v_{\Lambda}|H(z)\scr{G}_{t}|v_{\Lambda}}$.

\begin{thm}
Let $g_{t}(z)=\rho_{t}(z)+B_{t}$ be the SLE$(\kappa)$ and $h^{i}_{t}(\zeta)$ be the solutions of Eq.(\ref{eq:sde_internal_Heisenbeg}).
Then the following quantity is a local martingale:
\begin{equation}
	\braket{v_{\Lambda}|H(z)\scr{G}_{t}|v_{\Lambda}}
	=\lambda(H_{1}|H)\frac{\pr{\rho}_{t}(z)}{\rho_{t}(z)}-\sum_{i=1}^{\ell}(H_{i}|H)\del h^{i}_{t}(z).
\end{equation}
\end{thm}

We move on to derive the transformation rule for the Virasoro field $L(z)$.
The formula in Eq.(\ref{eq:virasoro_internal}) implies
\begin{equation}
	e^{-H\otimes h(\zeta)}L(z)e^{H\otimes h(\zeta)}=L(z)-\del h(z) H(z)+\frac{1}{2}(H|H)\del h(z)^{2}.
\end{equation}
Note that $\{H_{i}\}_{i=1}^{\ell}$ is an orthonormal basis, thus the corresponding currents $H_{i}(z)$ are mutually commutative.
This enables us to compute the quantity $\Theta_{t}^{-1}L(z)\Theta_{t}$ so that
\begin{equation}
	\Theta_{t}^{-1}L(z)\Theta_{t}=L(z)-\sum_{i=1}^{\ell}\del h^{i}_{t}(z)H_{i}(z)+\frac{1}{2}\sum_{i=1}^{\ell}\del h^{i}_{t}(z)^{2}.
\end{equation}
When we further take {{the}} adjoint by $Q(\rho_{t})^{-1}$ on it, we obtain
\begin{align}
	\scr{G}_{t}^{-1}L(z)\scr{G}_{t}=
	&L(\rho_{t}(z))\del\rho_{t}(z)^{2}-\sum_{i=1}^{\ell} \del h^{i}_{t}(z)\del \rho_{t}(z)H_{i}(\rho_{t}(z)) \notag \\
	&+\frac{c}{12}(S\rho_{t})(z)+\frac{1}{2}\sum_{i=1}^{\ell}\del h^{i}_{t}(z)^{2}. 
\end{align}
This relation again helps us {{to formulate}} a local martingale $\braket{v_{\Lambda}|L(z)\scr{G}_{t}|v_{\Lambda}}$ associated with the random processes $\rho_{t}(z)$ and $h^{i}_{t}(z)$.
\begin{thm}
Let $g_{t}(z)=\rho_{t}(z)+B_{t}$ be the SLE$(\kappa)$ and $h^{i}_{t}(\zeta)$ be the solutions of Eq.(\ref{eq:sde_internal_Heisenbeg}).
Then the following quantity is a local martingale:
\begin{align}
	\braket{v_{\Lambda}|L(z)\scr{G}_{t}|v_{\Lambda}}=
	h_{\Lambda}\left(\frac{\del \rho_{t}(z)}{\rho_{t}(z)}\right)^{2}-\lambda\del h^{1}_{t}(z) \frac{\del \rho_{t}(z)}{\rho_{t}(z)}
	+\frac{c}{12}(S\rho_{t})(z)+\frac{1}{2}\sum_{i=1}^{\ell}\del h^{i}_{t}(z)^{2}.
\end{align}
\end{thm}

Since on our representation space $L_{\mfrak{g}}(\Lambda,1)$ the Virasoro field is realized by using current fields,
the local martingale $\braket{v_{\Lambda}|L(z)\scr{G}_{t}|v_{\Lambda}}$ has another description.
From the transformation rule of a current field $H(z)$, its positive and negative power parts are transformed as
\begin{align}
	\scr{G}_{t}^{-1}H(z)_{+}\scr{G}_{t}=
	&\sum_{m\in\mbb{Z}}\mrm{Res}_{w}\frac{\del \rho_{t}(w)\rho_{t}(w)^{-m-1}}{w-z}H(m)
	-\mrm{Res}_{w}\frac{1}{w-z}\sum_{i=1}^{\ell}(H_{i}|H)\del h^{i}_{t}(w), \\
	\scr{G}_{t}^{-1}H(z)_{-}\scr{G}_{t}=
	&\sum_{m\in\mbb{Z}}\mrm{Res}_{w}\frac{\del \rho_{t}(w) \rho_{t}(w)^{-m-1}}{z-w}H(m)
	-\mrm{Res}_{w}\frac{1}{z-w}\sum_{i=1}^{\ell}(H_{i}|H)\del h^{i}_{t}(w).
\end{align}
Here rational functions $\frac{1}{z-w}$ and $\frac{1}{w-z}$ are expanded in regions $|z|>|w|$ and $|w|>|z|$, respectively.
We will use a similar convention in the following.
Thus the local martingale associated with the normal ordered product $\no{H(z)^{2}}$ is computed as
\begin{align}
	\braket{v_{\Lambda}|\no{H(z)^{2}}\scr{G}_{t}|v_{\Lambda}}=
	&(H|H)\mrm{Res}_{w}\left[\frac{\del\rho_{t}(w)}{w-z}\del_{z}\left(\frac{1}{\rho_{t}(w)-\rho_{t}(z)}\right)-\frac{\del\rho_{t}(w)}{z-w}\del_{z}\left(\frac{\rho_{t}(z)\rho_{t}(w)^{-1}}{\rho_{t}(z)-\rho_{t}(w)}\right)\right]  \notag \\
	&+(\lambda (H_{1}|H))^{2}\left(\frac{\del \rho_{t}(z)}{\rho_{t}(z)}\right)^{2}
	-2\lambda (H_{1}|H)\sum_{i=1}^{\ell}(H_{i}|H)\del h^{i}_{t}(z)\frac{\del \rho_{t}(z)}{\rho_{t}(z)}
\end{align}
This enables us to derive another form of the local martingale $\braket{v_{\Lambda}|L(z)\scr{G}|v_{\Lambda}}$ so that
\begin{align}
	\braket{v_{\Lambda}|L(z)\scr{G}_{t}|v_{\Lambda}}=
	&\frac{\ell}{2}\mrm{Res}_{w}\left[\frac{\del\rho_{t}(w)}{w-z}\del_{z}\left(\frac{1}{\rho_{t}(w)-\rho_{t}(z)}\right)-\frac{\del\rho_{t}(w)}{z-w}\del_{z}\left(\frac{\rho_{t}(z)\rho_{t}(w)^{-1}}{\rho_{t}(z)-\rho_{t}(w)}\right)\right]  \notag \\
	&+h_{\Lambda}\left(\frac{\del \rho_{t}(z)}{\rho_{t}(z)}\right)^{2}-\lambda \frac{\del\rho_{t}(z)}{\rho_{t}(z)}\del h^{1}_{t}(z).
\end{align}
Comparing this with the same quantity, which is seemingly different, derived previously,
we obtain an equality among random processes
\begin{align}
	&\frac{\ell}{2}\mrm{Res}_{w}\left[\frac{\del\rho_{t}(w)}{w-z}\del_{z}\left(\frac{1}{\rho_{t}(w)-\rho_{t}(z)}\right)-\frac{\del\rho_{t}(w)}{z-w}\del_{z}\left(\frac{\rho_{t}(z)\rho_{t}(w)^{-1}}{\rho_{t}(z)-\rho_{t}(w)}\right)\right] \notag \\
	&=\frac{c}{12}(S\rho_{t})(z)+\frac{1}{2}\sum_{i=1}^{\ell}(\del h^{i}_{t}(z))^{2}.
\end{align}

\subsection{{{Case when}} $\mfrak{g}=\mfrak{sl}_{2}$}
We {{formulate}} in this subsection several local martingales associated with {{SLEs}} with affine symmetry
that are generated by a local martingale $\scr{G}_{t}v_{\Lambda}$ in $\overline{L_{\mfrak{sl}_{2}}(\Lambda,k)}$.
We treat the case when $\Lambda$ is the fundamental weight of $\mfrak{sl}_{2}$ and $k=1$
and use the description $L(\Lambda)=\mbb{C}v_{\Lambda}\oplus \mbb{C}v_{-\Lambda}$ of the fundamental weight as in Prop. \ref{prop:annihilator_sl2}

Firstly we {{note transformation formulas}} for current fields $X(z)$ for $X=E,H,F$ under the adjoint action by $\scr{G}_{t}^{-1}$.
\begin{lem}
\label{lem:transform_current}
\begin{align}
	\scr{G}_{t}^{-1}E(z)\scr{G}_{t}=
	&e^{-2h_{t}(z)}\del \rho_{t}(z) E(\rho_{t}(z)) +e^{-2h_{t}(z)}f_{t}(z)\del \rho_{t}(z)H(\rho_{t}(z)) \notag \\
	&-e^{-2h_{t}(z)}f_{t}(z)^{2}\del\rho_{t}(z)F(\rho_{t}(z)) -k\del f_{t}(z),  \\
	\scr{G}_{t}^{-1}H(z)\scr{G}_{t}=
	&2e^{-2h_{t}(z)}e_{t}(z)\del\rho_{t}(z)E(\rho_{t}(z))+(1+2e^{-2h_{t}(z)}e_{t}(z)f_{t}(z))\del\rho_{t}(z)H(\rho_{t}(z)) \notag \\
	&-(2f_{t}(z)+2e^{-2h_{t}(z)}e_{t}(z)f_{t}(z)^{2})\del\rho_{t}(z) F(\rho_{t}(z)) \notag \\
	&-k(2\del h_{t}(z)+2e^{-2h_{t}(z)}e_{t}(z)\del f_{t}(z)),  \\
	\scr{G}_{t}^{-1}F(z)\scr{G}_{t}=
	&-e^{-2h_{t}(z)}e_{t}(z)^{2}\del \rho_{t}(z) E(\rho_{t}(z)) \notag \\
	&-(e_{t}(z)+e^{-2h_{t}(z)}e_{t}(z)^{2}f_{t}(z))\del \rho_{t}(z) H(\rho_{t}(z)) \notag \\
	&+(2e_{t}(z)f_{t}(z)+e^{-2h_{t}(z)}e_{t}(z)^{2}f_{t}(z)^{2})\del\rho_{t}(z) F(\rho_{t}(z)) \notag \\
	&+k(2e_{t}(z)\del f_{t}(z)+e^{-2h_{t}(z)}e_{t}(z)^{2}\del f_{t}(z)-\del e_{t}(z)).
\end{align}
\end{lem}
This will allow us to compute local martingales of the form $\braket{v_{\pm\Lambda}|X(z)\scr{G}_{t}|v_{\pm\Lambda}}$ for $X=E,H,F$.

\begin{thm}
\label{thm:local_martingales_affine_current}
Let $\Lambda$ be the fundamental weight of $\mfrak{sl}_{2}$,
and the fundamental representation of $\mfrak{sl}_{2}$ be described by $L(\Lambda)=\mbb{C}v_{\Lambda}\oplus \mbb{C}v_{-\Lambda}$ as in Prop. \ref{prop:annihilator_sl2}.
We assume that $\kappa$ and $\tau$ be positive real numbers satisfying the relation $\kappa+2\tau-4=0$.
For the SLE$(\kappa)$ $g_{t}(z)=\rho_{t}(z)+B_{t}$ and random processes $e_{t}(z)$, $h_{t}(z)$ and $f_{t}(z)$ satisfying
the SDEs in Prop. \ref{prop:sde_internal_sl2},
the following quantities are local martingales.

\begin{enumerate}
\item 	$X=E$.
		\begin{align}
			\braket{v_{\Lambda}|E(z)\scr{G}_{t}|v_{\Lambda}}&=e^{-2h_{t}(z)}f_{t}(z)\frac{\del \rho_{t}(z)}{\rho_{t}(z)}-\del f_{t}(z), \\
			\braket{v_{-\Lambda}|E(z)\scr{G}_{t}|v_{\Lambda}}&=-e^{-2h_{t}(z)}f_{t}(z)^{2}\frac{\del \rho_{t}(z)}{\rho_{t}(z)}, \\
			\braket{v_{\Lambda}|E(z)\scr{G}_{t}|v_{-\Lambda}}&=e^{-2h_{t}(z)}\frac{\del \rho_{t}(z)}{\rho_{t}(z)}, \\
			\braket{v_{-\Lambda}|E(z)\scr{G}_{t}|v_{-\Lambda}}&=-e^{-2h_{t}(z)}f_{t}(z)\frac{\del \rho_{t}(z)}{\rho_{t}(z)}-\del f_{t}(z).
		\end{align}

\item 	$X=H$.
		\begin{align}
			\braket{v_{\Lambda}|H(z)\scr{G}_{t}|v_{\Lambda}}=&(1+2e^{-2h_{t}(z)}e_{t}(z)f_{t}(z))\frac{\del \rho_{t}(z)}{\rho_{t}(z)}  \notag \\
										&-(2\del h_{t}(z)+2e^{-2h_{t}(z)}e_{t}(z)\del f_{t}(z)),  \\
			\braket{v_{-\Lambda}|H(z)\scr{G}_{t}|v_{\Lambda}}=&-(2f_{t}(z)+2e^{-2h_{t}(z)}e_{t}(z)f_{t}(z)^{2})\frac{\del \rho_{t}(z)}{\rho_{t}(z)}, \\
			\braket{v_{\Lambda}|H(z)\scr{G}_{t}|v_{-\Lambda}}=&2e^{-2h_{t}(z)}e_{t}(z)\frac{\del \rho_{t}(z)}{\rho_{t}(z)}, \\
			\braket{v_{-\Lambda}|H(z)\scr{G}_{t}|v_{-\Lambda}}=&-(1+2e^{-2h_{t}(z)}e_{t}(z)f_{t}(z))\frac{\del \rho_{t}(z)}{\rho_{t}(z)} \notag \\
										&-(2\del h_{t}(z)+2e^{-2h_{t}(z)}e_{t}(z)\del f_{t}(z)).
		\end{align}

\item 	$X=F$.
		\begin{align}
			\braket{v_{\Lambda}|F(z)\scr{G}_{t}|v_{\Lambda}}=&-(e_{t}(z)+e^{-2h_{t}(z)}e_{t}(z)^{2}f_{t}(z))\frac{\del \rho_{t}(z)}{\rho_{t}(z)} \notag \\
										&+(2e_{t}(z)\del f_{t}(z)+e^{-2h_{t}(z)}e_{t}(z)^{2}\del f_{t}(z)-\del e_{t}(z)),  \\
			\braket{v_{-\Lambda}|F(z)\scr{G}_{t}|v_{\Lambda}}=&(2e_{t}(z)f_{t}(z)+e^{-2h_{t}(z)}e_{t}(z)^{2}f_{t}(z)^{2})\frac{\del \rho_{t}(z)}{\rho_{t}(z)}, \\
			\braket{v_{\Lambda}|F(z)\scr{G}_{t}|v_{-\Lambda}}=&-e^{-2h_{t}(z)}e_{t}(z)^{2}\frac{\del \rho_{t}(z)}{\rho_{t}(z)}, \\
			\braket{v_{-\Lambda}|F(z)\scr{G}_{t}|v_{-\Lambda}}=&(e_{t}(z)+e^{-2h_{t}(z)}e_{t}(z)^{2}f_{t}(z))\frac{\del \rho_{t}(z)}{\rho_{t}(z)} \notag \\
										&+(2e_{t}(z)\del f_{t}(z)+e^{-2h_{t}(z)}e_{t}(z)^{2}\del f_{t}(z)-\del e_{t}(z)).
		\end{align}
\end{enumerate}
\end{thm}
\begin{proof}
By assumption, we have that $\scr{G}_{t}\ket{v_{\pm\Lambda}}$ are local {{martingales}} in $\overline{L_{\mfrak{sl}_{2}}(\Lambda,1)}$ from 
{{Prop. \ref{prop:martingale_affine} and Prop. \ref{prop:annihilator_sl2}}}.
Thus the quantities $\braket{u|\scr{G}_{t}|v_{\pm\Lambda}}$ are local martingales.
Noticing that $\bra{v_{\pm\Lambda}}\scr{G}_{t}=\bra{v_{\pm\Lambda}}$ and using the formula in Lemma \ref{lem:transform_current},
we obtain the desired result.
\end{proof}

We also compute the local martingales $\braket{v_{\pm\Lambda}|L(z)\scr{G}_{t}|v_{\pm\Lambda}}$ for the Virasoro field $L(z)$.
The Virasoro field is found to be transformed under the adjoint action by $\scr{G}_{t}^{-1}$ as follows.
\begin{lem}
\begin{align}
	\scr{G}_{t}^{-1}L(z)\scr{G}_{t}=
	&(\del\rho_{t}(z))^{2}L(\rho_{t}(z)) \notag \\
	&-e^{-2h_{t}(z)}\del e_{t}(z)\del \rho_{t}(z) E(\rho_{t}(z)) \notag \\
	&-(\del h_{t}(z)+e^{-2h_{t}(z)}f_{t}(z)\del e_{t}(z))\del \rho_{t}(z) H(\rho_{t}(z)) \notag \\
	&-(\del f_{t}(z)-2f_{t}(z)\del h_{t}(z)-e^{-2h_{t}(z)}f_{t}(z)^{2}\del e_{t}(z))\del \rho_{t}(z) F(\rho_{t}(z)) \notag \\
	&+k((\del h_{t}(z))^{2}+e^{-2h_{t}(z)}\del e_{t}(z) \del f_{t}(z))+\frac{c}{12}(S\rho_{t})(z). 
\end{align}
\end{lem}
\begin{thm}
Let $\Lambda$ be the fundamental weight of $\mfrak{sl}_{2}$,
and the fundamental representation of $\mfrak{sl}_{2}$ be described by $L(\Lambda)=\mbb{C}v_{\Lambda}\oplus \mbb{C}v_{-\Lambda}$ as in Prop. \ref{prop:annihilator_sl2}.
We assume that $\kappa$ and $\tau$ be positive real numbers satisfying the relation $\kappa+2\tau-4=0$.
For the SLE$(\kappa)$ $g_{t}(z)=\rho_{t}(z)+B_{t}$ and random processes $e_{t}(z)$, $h_{t}(z)$ and $f_{t}(z)$ satisfying
the SDEs in Prop. \ref{prop:sde_internal_sl2},
the following quantities are local martingales:
\begin{align}
	\braket{v_{\Lambda}|L(z)\scr{G}_{t}|v_{\Lambda}}=
	&\frac{1}{4}\left(\frac{\del\rho_{t}(z)}{\rho_{t}(z)}\right)^{2}-(\del h_{t}(z)+e^{-2h_{t}(z)}f_{t}(z)\del e_{t}(z))\frac{\del \rho_{t}(z)}{\rho_{t}(z)} \notag \\
	&+((\del h_{t}(z))^{2}+e^{-2h_{t}(z)}\del e_{t}(z) \del f_{t}(z))+\frac{1}{12}(S\rho_{t})(z),  \\
	\braket{v_{-\Lambda}|L(z)\scr{G}_{t}|v_{\Lambda}}=
	&-(\del f_{t}(z)-2f_{t}(z)\del h_{t}(z)-e^{-2h_{t}(z)}f_{t}(z)^{2}\del e_{t}(z))\frac{\del \rho_{t}(z)}{\rho_{t}(z)}, \\
	\braket{v_{\Lambda}|L(z)\scr{G}_{t}|v_{-\Lambda}}=
	&-e^{-2h_{t}(z)}\del e_{t}(z)\frac{\del \rho_{t}(z)}{\rho_{t}(z)}, \\
	\braket{v_{-\Lambda}|L(z)\scr{G}_{t}|v_{-\Lambda}}=
	&\frac{1}{4}\left(\frac{\del\rho_{t}(z)}{\rho_{t}(z)}\right)^{2}+(\del h_{t}(z)+e^{-2h_{t}(z)}f_{t}(z)\del e_{t}(z))\frac{\del \rho_{t}(z)}{\rho_{t}(z)} \notag \\
	&+((\del h_{t}(z))^{2}+e^{-2h_{t}(z)}\del e_{t}(z) \del f_{t}(z))+\frac{1}{12}(S\rho_{t})(z).
\end{align}
\end{thm}
\begin{proof}
The proof is analogous to {{that}} of Theorem \ref{thm:local_martingales_affine_current}.
We note that on $L_{\mfrak{sl}_{2}}(\Lambda,1)$, the central charge is $c=1$
and the conformal weight of the highest weight vector $v_{\Lambda}$ is $\frac{1}{4}$.
\end{proof}

%symmetry
\section{Symmetry of the space of local martingales}
\label{sect:affine_symmetry}
In the previous section, we saw that a local martingale $\scr{G}_{t}\ket{v_{\Lambda}}$ that takes its value in $\overline{L_{\mfrak{sl}_{2}}(\Lambda,k)}$
generates several local martingales.
We shall describe this phenomenon from a different point of view.

Let $\mcal{Y}(-,z)$ be an intertwining operator of type $\binom{L_{\mfrak{sl}_{2}}(\Lambda,k)}{L_{\mfrak{sl}_{2}}(\Lambda,k)\ \ L_{\mfrak{sl}_{2},k}}$.
Then for a vector $v\in L_{\mfrak{g}}(\Lambda,k)$, we have $\mcal{Y}(v,z)\ket{0}=e^{zL_{-1}}v$.
This implies that for a vector $v \in L(\Lambda)$ in the top space of $L_{\mfrak{g}}(\Lambda,k)$
that is annihilated by an operator of the form of Eq.(\ref{eq:annihilator_affine}),
\begin{equation}
	\label{eq:affine_martingale_vector}
	\scr{G}_{t}v=\Theta_{t}Q(g_{t})\mcal{Y}(v,B_{t})\ket{0}
\end{equation}
is a local martingale.

For a generic element in $\mrm{Aut}_{+}\mcal{O}\ltimes G_{+}(\mcal{O})$ and an intertwining operator $\mcal{Y}(-,z)$
of type $\binom{L_{\mfrak{g}}(\Lambda,k)}{L_{\mfrak{g}}(\Lambda,k)\ \ L_{\mfrak{g},k}}$, the quantity
\begin{equation}
	\mcal{M}_{u}=\braket{u|\scr{G}\mcal{Y}(-,x)|0} \in L(\Lambda)^{\ast}[g_{n+1},e_{n},h_{n},f_{n}|n<0][[x]]=:\mcal{F}_{\mrm{aff}}(\Lambda)
\end{equation}
for any vector $u\in L_{\mfrak{g}}(\Lambda,k)$ gives a local martingale when we evaluate $g_{n}$, $e_{n}$, $h_{n}$, $f_{n}$ at the SLE with internal degrees of freedom,
and $x$ at the Brownian motion of variance $\kappa$.
Thus we may find the space of local martingales as a subspace of $\mcal{F}_{\mrm{aff}}(\Lambda)$.
Since $u$ is arbitrarily taken, the quantity $\mcal{M}_{X(\ell)u}$ for $X\in\mfrak{sl}_{2}$ and $\ell\in\mbb{Z}$ has the same property.
Thus if we find a operator $\scr{X}_{\ell}$ such that $\mcal{M}_{X(\ell)u}=\scr{X}_{\ell}\mcal{M}_{u}$,
it can describe affine Lie algebra symmetry of a space of local martingales in $\mcal{F}_{\mrm{aff}}(\Lambda)$.
The derivation of the operators $\scr{X}_{\ell}$ is presented in Appendix \ref{sect:app_operator_X},
and we {{note the results here:}}
\begin{align}
	\scr{E}_{\ell}=
	&-\sum_{n\le -1}\mrm{Res}_{z}\mrm{Res}_{w}\frac{w^{-n-1}e^{2h(w)}e^{-2h(z)}z^{-\ell}\pr{g}(z)}{g(w)-g(z)}\frac{\del}{\del e_{n}}   \notag \\
	&-\sum_{n\le -1}\mrm{Res}_{z}\mrm{Res}_{w}\frac{w^{-n-1}e^{-2h(z)}(f(z)-f(w))z^{-\ell}\pr{g}(z)}{g(w)-g(z)}\frac{\del}{\del h_{n}} \notag \\
	&+\sum_{n\le -1}\mrm{Res}_{z}\mrm{Res}_{w}\frac{w^{-n-1}e^{-2h(z)}(f(z)-f(w))^{2}z^{-\ell}\pr{g}(z)}{g(w)-g(z)}\frac{\del}{\del f_{n}} \notag \\
	&+\mrm{Res}_{z}\frac{e^{-2h(z)}z^{-\ell}\pr{g}(z)}{g(z)-x}\pi(E) \notag \\
	&+\mrm{Res}_{z}\frac{e^{-2h(z)}f(z)z^{-\ell}\pr{g}(z)}{g(z)-x}\pi(H) \notag \\
	&-\mrm{Res}_{z}\frac{e^{-2h(z)}f(z)^{2}z^{-\ell}\pr{g}(z)}{g(z)-x}\pi(F) \notag \\
	&+k\mrm{Res}_{z}\del f(z)e^{-2h(z)}z^{-\ell}.
\end{align}
\begin{align}
	\scr{H}_{\ell}=
	&-2\sum_{n\le -1}\mrm{Res}_{z}\mrm{Res}_{w}\frac{w^{-n-1}e^{2h(w)}e^{-2h(z)}e(z)z^{-\ell}\pr{g}(z)}{g(w)-g(z)}\frac{\del}{\del e_{n}} \notag \\
	&-\sum_{n\le -1}\mrm{Res}_{z}\mrm{Res}_{w}\frac{w^{-n-1}(1+2e^{-2h(z)}(f(z)-f(w)))z^{-\ell}\pr{g}(z)}{g(w)-g(z)}\frac{\del}{\del h_{n}} \notag \\
	&-2\sum_{n\le -1}\mrm{Res}_{z}\mrm{Res}_{w}\frac{w^{-n-1}(f(w)-f(z)-e^{-2h(z)}e(z)(f(w)-f(z))^{2})z^{-\ell}\pr{g}(z)}{g(w)-g(z)}\frac{\del}{\del f_{n}} \notag \\
	&+2\mrm{Res}_{z}\frac{e^{-2h(z)}e(z)z^{-\ell}\pr{g}(z)}{g(z)-x}\pi(E) \notag \\
	&+\mrm{Res}_{z}\frac{(1+2e^{-2h(z)}e(z)f(z))z^{-\ell}\pr{g}(z)}{g(z)-x}\pi(H) \notag \\
	&-2\mrm{Res}_{z}\frac{(1+e^{-2h(z)}e(z)f(z))f(z)z^{-\ell}\pr{g}(z)}{g(z)-x}\pi(F) \notag \\
	&+2k\mrm{Res}_{z}(\del h(z)-\del f(z)e^{-2h(z)}e(z))z^{-\ell}.
\end{align}
\begin{align}
	\scr{F}_{\ell}=
	&\sum_{n\le -1}\mrm{Res}_{z}\mrm{Res}_{w}\frac{w^{-n-1}e^{2h(w)}e^{-2h(z)}e(z)^{2}z^{-\ell}\pr{g}(z)}{g(w)-g(z)}\frac{\del}{\del e_{n}} \notag \\
	&-\sum_{n\le -1}\mrm{Res}_{z}\mrm{Res}_{w}\frac{w^{-n-1}(1+e^{-2h(z)}e(z)(f(w)-f(z)))e(z)z^{-\ell}\pr{g}(z)}{g(w)-g(z)}\frac{\del}{\del h_{n}} \notag \\
	&-\sum_{n\le -1}\mrm{Res}_{z}\mrm{Res}_{w}w^{-n-1}\Biggl[\frac{e^{2h(z)}+2e(z)(f(z)-f(w))}{g(w)-g(z)} \notag \\
	&\hspace{120pt}+\frac{e^{-2h(z)}e(z)^{2}(f(z)-f(w))^{2}}{g(w)-g(z)}\Biggr]z^{-\ell}\pr{g}(z)\frac{\del}{\del f_{n}} \notag \\
	&-\mrm{Res}_{z}\frac{e^{-2h(z)}e(z)^{2}z^{-\ell}\pr{g}(z)}{g(z)-x}\pi(E) \notag \\
	&-\mrm{Res}_{z}\frac{(1+e^{-2h(z)}e(z)f(z))e(z)z^{-\ell}\pr{g}(z)}{g(z)-x}\pi(H) \notag \\
	&+\mrm{Res}_{z}\frac{(e^{2h(z)}+2e(z)f(z)+e^{-2h(z)}e(z)^{2}f(z)^{2})z^{-\ell}\pr{g}(z)}{g(z)-x}\pi (F) \notag \\
	&-\mrm{Res}_{z}(2\del h(z)e(z)-\del e(z)+\del f(z)e^{-2h(z)}e(z)^{2})z^{-\ell}.
\end{align}
Here the representation $\pi:\mfrak{sl}_{2}\to \mrm{End}(L(\Lambda)^{\ast})$ is defined by
$(\pi(X)\phi)(v)=-\phi(Xv)$ for $X\in\mfrak{sl}_{2}$, $\phi\in L(\Lambda)^{\ast}$ and $v\in L(\Lambda)$.

We can also derive operators $\scr{L}_{\ell}$ that associate with the action of the Virasoro algebra such that
$\mcal{M}_{L_{\ell}u}=\scr{L}_{\ell}\mcal{M}_{u}$.
{{The}} detailed derivation is {{conducted in}} Appendix \ref{sect:app_operator_X}, {{but}} it yields
\begin{align}
	\scr{L}_{\ell}=
	&-\sum_{n\le 0}\mrm{Res}_{z}\mrm{Res}_{w}\frac{z^{-\ell+1}w^{-n-1}\pr{g}(z)^{2}}{g(w)-g(z)}\frac{\del}{\del g_{n}}  \notag \\
	&-\sum_{n\le -1}\mrm{Res}_{z}\mrm{Res}_{w}\frac{z^{-\ell+1}w^{-n-1}e^{2h(w)}e^{-2h(z)}\del e(z)\pr{g}(z)}{g(w)-g(z)}\frac{\del}{\del e_{n}} \notag \\
	&-\sum_{n\le -1}\mrm{Res}_{z}\mrm{Res}_{w}\frac{z^{-\ell+1}w^{-n-1}(\del h(z)+e^{-2h(z)}\del e(z)(f(z)-f(w)))\pr{g}(z)}{g(w)-g(z)}\frac{\del}{\del h_{n}} \notag \\
	&-\sum_{n\le -1}\mrm{Res}_{z}\mrm{Res}_{w}z^{-\ell+1}w^{-n-1}\Biggl[\frac{\del f(z)-2\del h(z)(f(z)-f(w))}{g(w)-g(z)} \notag \\
	&\hspace{150pt}-\frac{e^{-2h(z)}\del e(z)(f(z)-f(w))^{2}}{g(w)-g(z)}\Biggr]\pr{g}(z)\frac{\del}{\del f_{n}} \notag \\
	&+\mrm{Res}_{z}z^{-\ell+1}\pr{g}(z)^{2}\left(\frac{h}{(g(z)-x)^{2}}+\frac{1}{g(z)-x}\frac{\del}{\del x}\right) \notag \\
	&+\mrm{Res}_{z}\frac{z^{-\ell+1}e^{-2h(z)}\del e(z)\pr{g}(z)}{g(z)-x}\pi(E) \notag \\
	&+\mrm{Res}_{z}\frac{z^{-\ell+1}(\del h(z)+e^{-2h(z)}f(z)\del e(z))\pr{g}(z)}{g(z)-x}\pi(H) \notag \\
	&+\mrm{Res}_{z}\frac{z^{-\ell+1}(\del f(z)-2f(z)\del h(z)-e^{-2h(z)}f(z)^{2}\del e(z))\pr{g}(z)}{g(z)-x}\pi(F) \notag \\
	&+\mrm{Res}_{z}z^{-\ell+1}\left(\frac{c}{12}(Sg)(z)+k(\del h(z)^{2}+e^{-2h(z)}\del f(z)\del e(z))\right).
\end{align}

For a vector $v\in L(\Lambda)$ in the top space of $L_{\mfrak{sl}_{2}}(\Lambda,k)$,
the corresponding local martingale $\mcal{M}_{v}$ is a constant function in $x$ that takes value $\braket{v|-}\in L(\Lambda)^{\ast}$.
Applying the operators $\scr{X}_{\ell}$ on elements $\mcal{M}_{v}$ for $v\in L(\Lambda)$,
we obtain all local martingales that are generated by a random process $\scr{G}_{t}$ on $\mrm{Aut}_{+}\mcal{O}\ltimes G_{+}(\mcal{O})$.

\begin{thm}
Assume that we have an operator of the form in Eq.(\ref{eq:annihilator_affine}) that annihilates the highest weight vector of $L_{\mfrak{sl}_{2}}(\Lambda,k)$.
Let $\mcal{U}$ be the subspace of $\mcal{F}_{\mrm{aff}}(\Lambda)$ that is obtained by applying operators $\scr{X}_{\ell}$ for $\scr{X}=\scr{E},\scr{H},\scr{F}$ and $\ell\in\mbb{Z}$
to elements of the form $\braket{u|-}\in L(\Lambda)^{\ast}$ for $u\in\L(\Lambda)$.
Then an element of $\mcal{U}$
gives a local martingale when the SLE$(\kappa)$, the solution of the SDEs in Prop. \ref{prop:sde_internal_sl2},
and the Brownian motion of variance $\kappa$ are substituted.
Namely, for an element $f(g_{n},e_{n},h_{n},f_{n},x)\in \mcal{U}$,
\begin{equation}
	f(g_{n}(t),e_{n}(t),h_{n}(t),f_{n}(t),B_{t})(u)
\end{equation}
is a local martingale for an arbitrary $u\in L(\Lambda)$.
Here
\begin{equation}
	g_{t}(z)=z+\sum_{n\le 0}g_{n}(t)z^{n}
\end{equation}
is the SLE$(\kappa)$ and
\begin{align}
	e_{t}(z)&=\sum_{n<0}e_{n}(t)z^{n}, \\
	h_{t}(z)&=\sum_{n<0}h_{n}(t)z^{n}, \\
	f_{t}(z)&=\sum_{n<0}f_{n}(t)z^{n}
\end{align}
satisfy the SDEs in {{Prop.}} \ref{prop:sde_internal_sl2}.
\end{thm}

%conclusion
\section{Conclusion}
In this paper, we have established the group theoretical formulation of {{SLEs}} corresponding to affine Lie algebras following previous work by the {{present}} author.\cite{SK2017}
As illustrated in Sect.\ref{sect:group_theoretical}, SLE/CFT correspondence
in the sense of Bauer and Bernard \cite{BauerBernard2002,BauerBernard2003a,BauerBernard2003b}
{{allowed}} us to compute local martingales associated with {{SLEs}} from representations of the Virasoro {{algebra}}.
Our achievement {{was}} to generalize this notion of SLE/CFT correspondence to connection between
SDEs and representations of affine Lie algebra.
Our strategy {{was}} to extend a random process on an {{infinite-dimensional}} Lie group $\mrm{Aut}_{+}\mcal{O}$
that {{was}} naturally connected to {{SLEs}} associated with the Virasoro algebra
to a random process on a larger group $\mrm{Aut}_{+}\mcal{O}\ltimes G_{+}(\mcal{O})$,
which {{was}} introduced in Sect.\ref{sect:internal_symmetry}.
The SDE for a random process on such an {{infinite-dimensional Lie group was given}} in Sect.\ref{sect:random_process}
based on consideration {{of}} an annihilating operator of a highest weight vector.
It is significant that the Virasoro module structure on a representation of an affine Lie algebra {{was}} introduced via the Segal-Sugawara construction.
Note that the resulting SDEs have already appeared
in the correlation function formulation\cite{BettelheimGruzbergLudwigWiegmann2005, AlekseevBytskoIzyurov2011} of {{SLEs}} corresponding to WZW theory
in an equivalent form, but we {{gave}} another natural derivation of it from a random process on an {{infinite-dimensional}} Lie group.
We also {{constructed}} the random process in the case when the underlying {{finite-dimensional Lie algebra was}} commutative and $\mfrak{sl}_{2}$.
%{\color{red}{
A significant achievement in the present paper was the derivation of SDEs in Prop.\ref{prop:sde_internal_sl2},
which gives a rigorous formulation of the random process along the internal space.
Thus it paves the way for further studies of SLE with internal degrees of freedom in probability theory.
%}}
Such a construction made it possible in Sect.\ref{sect:martingale} to {{formulate}} several local martingales associated with {{SLEs}} from computation on
a representation of an affine Lie algebra.
We also {{revealed}} an affine $\mfrak{sl}_{2}$ symmetry of a space of local martingales in Sect.\ref{sect:affine_symmetry},
which {{was}} again possible due to the construction in Sect.\ref{sect:random_process}.
It is clear that the content of Sect.\ref{sect:affine_symmetry} can be extended to other affine Lie algebras in principle,
but it will be harder to {{formulate}} operators defining the action.

Let us discuss other possibility of a random process on $\mrm{Aut}_{+}\mcal{O}\ltimes G_{+}(\mcal{O})$.
In Sect.\ref{sect:random_process}, we considered a random process $\scr{G}_{t}$ on an {{infinite-dimensional}} Lie group $\mrm{Aut}_{+}\mcal{O}\ltimes G_{+}(\mcal{O})$,
of which the $dt$ term in its increment is an annihilating operator
\begin{equation}
	-2L_{-2}+\frac{\kappa}{2}L_{-1}^{2}+\frac{\tau}{2}\sum_{i=1}^{\dim\mfrak{g}}X_{i}(-1)^{2}
\end{equation}
of the highest weight vector.
This annihilator is chosen by the following principle.
Firstly, our construction should derive the ordinary SLE in the coordinate transformation part,
which forces an annihilator to have the part $-2L_{-2}+\frac{\kappa}{2}L_{-1}^{2}$.
Secondly, the operator of the above form indeed annihilates the vacuum vector due to the Segal-Sugawara construction of the Virasoro generators.
The third term of the annihilator has room for generalization, which we shall discuss.
We can allow the variance $\tau$ to depend on $i${{; namely,}} an annihilator of the form
\begin{equation}
	-2L_{-2}+\frac{\kappa}{2}L_{-1}^{2}+\frac{1}{2}\sum_{i=1}^{\dim\mfrak{g}}\tau_{i}X_{i}(-1)^{2}
\end{equation}
can be considered.
We can also deform the annihilator by adding a term like $X(-2)$ for $X\in\mfrak{g}$.
Such a deformation will be inevitable if we twist the Virasoro generators by a derivative of a current field.
The problem {{of}} whether annihilators generalized in these ways indeed annihilate the highest weight vector,
of course, requires {{individual}} case-by-case investigation.

A possible application of our construction of SLEs corresponding to affine Lie algebras is to derive {{a}} generalization of Cardy's formula.
In {{the case of}} Virasoro algebra, SLE/CFT correspondence derives Cardy's formula.\cite{BauerBernard2003b}
We shall discuss {{the possibility of}} generalization of Cardy's formula.
This work will be {{two-fold. One aim is }} to find an appropriate scaling limit of a model of statistical physics in which a kind of cluster interface is described {{by the}} SLE derived in our formulation.
An important point to be considered is that our SLE with internal degrees of freedom describes a stochastic deformation of $G$-bundles
which requires us to find a scaling limit that captures such internal degrees of freedom as well as a cluster interface itself.
The other aim is to relate an object like
\begin{equation}
	\braket{u|\mcal{Y}(A_{1},z_{1})\cdots \mcal{Y}(A_{n},z_{n})\scr{G}_{t}|v_{\Lambda}}
\end{equation}
to the defining function of an event associated with SLEs with internal degrees of freedom derived in this paper.
Here $\mcal{Y}(-,z)$ is an intertwining operator and $\scr{G}_{t}\ket{v_{\Lambda}}$ is a representation-space-valued local martingale
constructed in this paper.
If such a discussion is possible, the probability of the event is computed as the expectation value of the above quantity,
which is time independent and thus reduces to a purely algebraic quantity
\begin{equation}
	\braket{u|\mcal{Y}(A_{1},z_{1})\cdots \mcal{Y}(A_{n},z_{n})|v_{\Lambda}}
\end{equation}
and which may be computed.

It is natural to seek other examples of generalization of SLEs involving more general internal symmetry.
In a forthcoming paper,\cite{SK2018b} we will construct SLEs for which the internal symmetry is described by an affine Lie superalgebra.
Since the Segal-Sugawara construction also works for a twisted affine Lie algebra, 
a parallel construction to ours presented in this paper will be possible for a twisted affine Lie algebra.
We consider the case that internal symmetry is encoded in a more complicated Lie algebra to be nontrivial.
We can associate with a VOA a Lie algebra called a current Lie algebra.
A Lie subalgebra of a current Lie algebra possibly describes an internal symmetry in the terminology of the book by Frenkel and Ben-Zvi.\cite{FrenkelBen-Zvi2004}
For example, the current Lie algebra of an affine VOA has the corresponding affine Lie algebra as a subalgebra,
and this is the internal symmetry we treated in this paper.
However it is not always possible to take such a {\it good} Lie subalgebra for a given VOA,
and it is nontrivial whether one can construct SLEs with internal degrees of freedom in such situations
when we do not know they have a good Lie subalgebra.

\subsection*{acknowledgements}
The author is grateful to K. Sakai for leading him to this field of research,
and to R. Sato for discussions and helpful advice.
He also would like to thank R. Friedrich for useful comments on the first version of the manuscript
and suggestions for suitable references.
This work was supported by a Grant-in-Aid for JSPS Fellows (Grant No. 17J09658).

\appendix

%app_vertex_operator_algebra
\section{Remarks on VOAs}
\label{sect:app_voa}
In this appendix, we recall the notion of vertex (operator) algebras which is useful in the present paper.
Detailed expositions of the theory of vertex (operator) algebras can be found in literatures.\cite{Kac1998, FrenkelBen-Zvi2004}
The appendix of the book by Iohara and Koga\cite{IoharaKoga2011} is also useful.

\subsection{Definition of vertex algebras, modules and intertwining operators}
Let $V$ be a vector space.
A field on $V$ is a series $a(z)=\sum_{n\in\mathbb{Z}}a_{(n)}z^{-n-1}$ in a formal variable $z$ with coefficients $a_{(n)}$ being in $\mathrm{End}(V)$
such that for any $v\in V$ we have $a_{(n)}v=0$ for $n \gg 0$.
Equivalently, a field is a linear map from $V$ to $V((z))=V[[z]][z^{-1}]$.
We let the space of fields be denoted by $\mrm{Fie}(V):=\mrm{Hom}_{\mbb{C}}(V,V((z)))$.

\begin{defn}[Vertex algebra]
A vertex algebra is a quadruple $(V,\ket{0},T,Y)$ of a vector space $V$, a distinguished vector $\ket{0}\in V$,
an operator $T\in\mrm{End}(V)$, and a linear operator $Y\in\mrm{Hom}(V,\mrm{Fie}(V))$,
on which the following axioms are imposed:
\begin{description}
	 \item[(VA1)] 	{\bf (translation covariance)}
		  	\begin{equation}
		  		\label{eq:translation_covariance}
		  		[T, Y(a,z)]=\del Y(a,z)
		  	\end{equation}
	 \item[(VA2)] 	{\bf (vacuum axioms)}
	 		\begin{align}
				\label{eq:vacuum_axiom}
	 			T\ket{0}&=0,&
	 			Y(\ket{0},z)&=\mrm{Id}_{V},&
	 			Y(a,z)\ket{0}\big|_{z=0}&=a.
			\end{align}
	 \item[(VA3)]{\bf (locality)}
		  	\begin{equation}
		  		(z-w)^{N}[Y(a,z),Y(b,w)]=0,\ \ N\gg 0.
		  	\end{equation}
\end{description}
Here we have denoted the image of $a\in V$ via $Y$ by $Y(a,z)$.
\end{defn}

We often denote a vertex algebra $(V,\ket{0},T,Y)$ simply by $V$.
We also often expand a field $Y(A,z)$ so that $Y(A,z)=\sum_{n\in\mbb{Z}}A_{(n)}z^{-n-1}$.

\begin{defn}
Let $V$ be a vertex algebra and $S\subset V$ be a subset.
We say that $V$ is generated by $S$ if $V$ is spanned by vectors of the form $A^{1}_{(-i_{1})}\cdots A^{n}_{(-i_{n})}\ket{0}$
for $A^{1},\cdots ,A^{n}\in S$, $i_{1},\cdots i_{n}\in\mbb{Z}_{\ge 1}$ and $n\ge 0$.
\end{defn}

\begin{defn}
A vertex algebra $V$ is said to be $\mbb{Z}$-graded if it admits a $\mbb{Z}$-gradation $V=\bigoplus_{n\in\mbb{Z}}V_{n}$ such that
$\ket{0}\in V_{0}$, $TV_{n}\subset V_{n+1}$, and $(V_{h})_{(n)}(V_{\pr{h}})\subset V_{h+\pr{h}-n-1}$ for any $h,\pr{h},n\in\mbb{Z}$.
We say that a vector in $V_{h}$ has conformal weight $h$.
\end{defn}

\begin{defn}
A vector $\omega\in V$ is a conformal vector of central charge $c$ if the coefficients of $Y(\omega,z)=\sum_{n\in\mbb{Z}}L_{n}z^{-n-2}$
define a representation of the Virasoro algebra of central charge $c$, or explicitly satisfy the commutation relation
\begin{equation}
	[L_{m},L_{n}]=(m-n)L_{m+n}+\frac{c}{12}(m^{3}-m)\delta_{m+n,0},
\end{equation}
with $L_{-1}=T$, and $L_{0}$ is diagonalizable on $V$.
A vertex algebra endowed with a conformal vector $\omega$ is called a conformal vertex algebra of central charge $c$.
The field $Y(\omega,z)$ is called a Virasoro field of the conformal vertex algebra $V$.
\end{defn}

\begin{defn}[Vertex operator algebra]
A $\mbb{Z}$-graded conformal vertex algebra $(V=\bigoplus_{n\in\mbb{Z}}V_{n},\omega)$ is called a VOA if
\begin{itemize}
\item 	$L_{0}|_{V_{n}}=n\mrm{id}_{V_{n}}$ for all $n\in \mbb{Z}$.
\item 	$\dim V_{n} <\infty$ for all $n\in\mbb{Z}$.
\item 	There exists $N\in \mbb{Z}$ such that $V_{n}=\{0\}$ for $n<N$.
\end{itemize}
\end{defn}

\begin{defn}
Let $(V,\ket{0}, T, Y, \omega)$ be a VOA.
A weak $V$-module is a pair $(M, Y^{M})$ of a vector space $M$
and a linear map $Y^{M}:V\to \mrm{End}(M)[[z,z^{-1}]]$
satisfying the following conditions:
\begin{itemize}
\item 	$Y^{M}(\ket{0},z)=\mrm{id}_{M}$.
\item 	For arbitrary $A\in V$ and $v \in M$,
		\begin{equation*}
			Y^{M}(A,z)v\in M((z)).
		\end{equation*}
\item 	For arbitrary $A, B\in V$ and $m,n \in\mbb{Z}$,
		\begin{align*}
			&\mrm{Res}_{z-w}Y^{M}(Y(A,z-w)B,w)i_{w,z-w}z^{m}(z-w)^{n} \\
			&=\mrm{Res}_{z}Y^{M}(A,z)Y^{M}(B,w)i_{z,w}z^{m}(z-w)^{n} \\
			&\hspace{10pt}-\mrm{Res}_{z}Y^{M}(B,w)Y^{M}(A,z)i_{w,z}z^{m}(z-w)^{n}.
		\end{align*}
		Here for a rational function $R(z,w)$ in two variables $z$ and $w$ possibly with poles at $z=0$, $w=0$, and $z=w$,
		we denote its expansion in the region $|z|>|w|$ by $i_{z,w}R(z,w)\in\mbb{C}[[z,z^{-1},w,w^{-1}]]$.
\end{itemize}
For a weak $V$-module $(M,Y^{M})$, the image of $A\in V$ by $Y^{M}$ is expressed as
\begin{equation}
	Y^{M}(A,z)=\sum_{n\in\mbb{Z}}A^{M}_{(n)}z^{-n-1}
\end{equation}
with $A^{M}_{(n)}\in\mrm{End}(M)$.
\end{defn}

If $A\in V$ has the conformal weight $\Delta$, it is convenient to expand $Y^{M}(A,z)$ as
\begin{equation}
	Y^{M}(A,z)=\sum_{n\in\mbb{Z}}A^{M}_{n}z^{-n-\Delta}
\end{equation}
so that $\deg A^{M}_{n}=-n$.

\begin{defn}
Let $V$ be a VOA and $\omega\in V$ be the conformal vector of $V$.
An ordinary $V$-module is a weak $V$-module $M$ such that
\begin{itemize}
\item 	$L^{M}_{0}$ in the expansion
		\begin{equation*}
			Y^{M}(\omega,z)=\sum_{n\in\mbb{Z}}L^{M}_{n}z^{-n-2}
		\end{equation*}
		is diagonalizable on $M$.
\item 	In the $L^{M}_{0}$-eigenspace decomposition
		\begin{equation*}
			M=\bigoplus_{\lambda\in\mbb{C}}M_{\lambda},
		\end{equation*}
		$\dim M_{\lambda}<\infty$ for all $\lambda\in\mbb{C}$.
		Moreover, for arbitrary $\lambda\in\mbb{C}$, $M_{\lambda-n}=0$ for $n\gg 0$.
\end{itemize}
\end{defn}

\begin{defn}
Let $M^{1}$, $M^{2}$ and $M^{3}$ be $V$-modules.
An intertwining operator of type $\binom{M_{1}}{M_{2}\ \ M_{3}}$ is a linear operator
\begin{equation}
	\mcal{Y}(-,z):M^{1}\to \mrm{Hom}(M^{2},M^{3})z^{K}:=\left\{\sum_{a \in K}v_{a}z^{a}\Bigg|v_{\alpha}\in\mrm{Hom}(M^{2},M^{3})\right\},
\end{equation}
where $K=\bigcup_{i}(\alpha_{i}+\mbb{Z})$ with finitely many $\alpha_{i}\in\mbb{C}$ being chosen associated with $M^{1}$, $M^{2}$ and $M^{3}$
with the following conditions imposed:
\begin{itemize}
\item
	For any $A\in V$, $v\in M^{1}$ and $m,n\in\mbb{Z}$ we have
	\begin{align*}
		&\mrm{Res}_{z-w}\mcal{Y}(Y^{M^{1}}(A,z-w)v,w)i_{w,z-w}z^{m}(z-w)^{n} \\
		&=\mrm{Res}_{z}Y^{M^{3}}(A,z)\mcal{Y}(v,w)i_{z,w}z^{m}(z-w)^{n} \\
		&\hspace{10pt}-\mrm{Res}_{z}\mcal{Y}(v,w)Y^{M^{2}}(A,z)i_{w,z}z^{m}(z-w)^{n}.
	\end{align*}
\item
	For any $v\in M^{1}$, we have
	\begin{equation}
		\mcal{Y}(L_{-1}v,z)=\frac{d}{dz}\mcal{Y}(v,z).
	\end{equation}
\end{itemize}
\end{defn}

\subsection{Examples}
\subsubsection{Virasoro vertex algebra}
In Sect. \ref{sect:group_theoretical}, we have introduced two types of representations of the Virasoro algebra,
Verma modules and their simple quotients.
We can also consider intermediate objects in the theory of VOAs.
The Verma module $M(c,0)$ of highest weight $(c,0)$ has a submodule generated by $L_{-1}{\bf 1}_{c,0}$.
Then the universal Virasoro VOA of central charge $c$ is defined by
\begin{equation}
	V_{c}:=M(c,0)/U(\mrm{Vir_{<0}})L_{-1}{\bf 1}_{c,0}.
\end{equation}
Now we prepare the components of the vertex algebra structure on $V_{c}$.
\begin{itemize}
\item 	$\ket{0}={\bf 1}_{c,0}$,
\item 	$T=L_{-1}$,
\item 	A single generator $\omega=L_{-2}\ket{0}$ with the corresponding field $L(z)=\sum_{n\in\mbb{Z}}L_{n}z^{-n-2}$.
\end{itemize}
From these data, we construct a vertex algebra structure on $V_{c}$ by
\begin{equation}
	Y(L_{j_{1}}\cdots L_{j_{k}}\ket{0},z)=\no{\del^{(-j_{1}-2)}L(z)\cdots \del^{(-j_{k}-2)}L(z)},
\end{equation}
with $L(z)=Y(\omega,z)$.
Moreover, $V$ is $\mbb{Z}$-graded by
\begin{equation}
	\deg (L_{j_{1}}\cdots L_{j_{k}}\ket{0})=-\sum_{i=1}^{k}j_{i}.
\end{equation}
Then $\omega\in V_{2}$ and $\deg L_{n}=-n$, implying $V$ is equipped with a $\mbb{Z}$-graded vertex algebra.
We also see $\omega$ is a conformal vector, and $V$ is a VOA.
It is obvious that the maximal proper submodule of $V_{c}$ as a $\mrm{Vir}$-module is a vertex subalgebra.
Thus the irreducible representation $L(c,0)$ of the Virasoro algebra also carries a vertex algebra structure
and we denote this vertex algebra by $L_{c}$.

Modules over $L_{c}$ are realized as simple highest weight representations $L(c,h)$ of the same central charge.
Note that an arbitrary simple representation of the Virasoro algebra is not necessarily a module over $L_{c}$,
since nontrivial relations may be imposed on the VOA $L_{c}$.
For instance, if the central charge is given by
\begin{equation}
	c=c_{p,q}=1-\frac{6(p-q)^{2}}{pq}
\end{equation}
with coprime integers $p$ and $q$ greater than or equal to $2$,
the corresponding Virasoro VOA is rational and its simple modules are exhausted by $L(c_{p,q},h_{p,q:r,s})$ with
\begin{equation}
	h_{p,q:r,s}=\frac{(rp-sq)^{2}-(p-q)^{2}}{4pq},\ \ 0<r<q,\ 0<s<p.
\end{equation}

\subsubsection{Affine vertex algebra}
Representations $\what{L(0)}_{k}$ and $L_{\mfrak{g},k}$ of an affine Lie algebra $\what{\mfrak{g}}$ introduced in Sect.\ref{sect:rep_aff_alg}
are also equipped with VOA structure by the following data:
\begin{itemize}
\item 	$\ket{0}=v_{0}$,
\item 	$T=L_{-1}$,
\item 	Generators $X_{a}(-1)\ket{0}$ for $a=1,\cdots,\dim\mfrak{g}$ with the corresponding fields $X_{a}(z)=\sum_{n\in\mbb{Z}}X_{a}(n)z^{-n-1}$.
\end{itemize}
Modules over an affine VOA are realized as $L_{\mfrak{g}}(\Lambda,k)$ of the same level $k$,
but again all these representations of the affine Lie algebra are not necessarily modules over the simple VOA $L_{\mfrak{g},k}$.
Indeed, we have the following example.

\begin{thm}[Frenkel-Zhu \cite{FrenkelZhu1992}]
Let $\mfrak{g}$ be a finite-dimensional simple Lie algebra and $k\in\mbb{Z}_{>0}$.
The simple $L_{\mfrak{g},k}$-modules are exhausted by $L_{\mfrak{g}}(\Lambda,k)$ with $\Lambda\in P_{+}^{k}$,
where $P_{+}^{k}$ is the set of dominant weights of level $k$ defined by
\begin{equation}
	P_{+}^{k}=\{\Lambda\in P_{+}|(\theta|\Lambda)\le k\}.
\end{equation}
\end{thm}

\subsubsection{Lattice vertex algebra}
Let $L$ be a nondegenerate even lattice of rank $\ell$; namely, it is a free $\mbb{Z}$-module of rank $\ell$
endowed with a nondegenerate $\mbb{Z}$-bilinear form $(\cdot|\cdot):L\times L\to \mbb{Z}$,
such that $(\alpha|\alpha)\in 2\mbb{Z}$ for $\alpha\in L$.
There uniquely exists a cohomology class $[\epsilon]\in H^{2}(L,\mbb{C}^{\times})$ satisfying
\begin{align}
	\label{eq:lattice_2_cocycle_1}
	\epsilon(\alpha,0)&=\epsilon(0,\alpha)=1,\\
	\label{eq:lattice_2_cocycle_2}
	\epsilon(\alpha,\beta)&=(-1)^{(\alpha|\beta)+|\alpha|^{2}|\beta|^{2}}\epsilon(\beta,\alpha)
\end{align}
for $\alpha,\beta\in L$.
Here we denote $|\alpha|^{2}=(\alpha|\alpha)$.
Notice that conditions Eq. (\ref{eq:lattice_2_cocycle_1}) and Eq. (\ref{eq:lattice_2_cocycle_2}) 
are independent of the choice of a representative $\epsilon$ of $[\epsilon]$.
It can be shown that we can choose a 2-cocycle $\epsilon\in[\epsilon]$ so that
it takes values in $\{\pm 1\}$. (See Remark 5.5a in the booklet by Kac.\cite{Kac1998})
We let $\epsilon$ be such a 2-cocycle in the following.
Let $\mbb{C}_{\epsilon}[L]$ be the $\epsilon$-twisted group algebra of $L$, which is
\begin{equation}
	\mbb{C}_{\epsilon}[L]=\bigoplus_{\alpha\in L}\mbb{C}e^{\alpha}
\end{equation}
as a vector space with multiplication defined by
\begin{equation}
	e^{\alpha}e^{\beta}=\epsilon(\alpha,\beta)e^{\alpha+\beta}
\end{equation}
for $\alpha,\beta\in L$.

We set $\mfrak{h}=\mbb{C}\otimes_{\mbb{Z}}L$ and extend the symmetric $\mbb{Z}$-bilinear form $(\cdot|\cdot)$ on $L$
to a symmetric $\mbb{C}$-bilinear form on $\mfrak{h}$.
Then we obtain the corresponding Heisenberg algebra $\widehat{\mfrak{h}}$ and its vacuum representation $L_{\mfrak{h}}(0,1)$ of level $1$.
The lattice vertex algebra $V_{L}$ associated to $L$ is
\begin{equation}
	V_{L}=L_{\mfrak{h}}(0,1)\otimes \mbb{C}_{\epsilon}[L]
\end{equation}
as a vector space.
We define the action of $\widehat{\mfrak{h}}$ on $V_{L}$ by
\begin{equation}
	H(m).(s\otimes e^{\alpha}):=(H(m)+\delta_{m,0}(H|\alpha))s\otimes e^{\alpha}
\end{equation}
for $H\in\mfrak{h}$, $m\in\mbb{Z}$, $s\in L_{\mfrak{h}}(0,1)$, and $\alpha\in L$.
We also define the action of $\mbb{C}_{\epsilon}[L]$ on $V_{L}$ by
\begin{equation}
	e^{\beta}.(s\otimes e^{\alpha}):=\epsilon(\beta,\alpha)s\otimes e^{\alpha+\beta}
\end{equation}
for $\alpha,\beta\in L$ and $s\in L_{\mfrak{h}}(0,1)$.
The lattice vertex algebra is generated by vectors $H(-1)\ket{0}\otimes e^{0}$ with $H\in \mfrak{h}$ and $\ket{0}\otimes e^{\alpha}$ with $\alpha\in L$,
of which the corresponding fields are given by
\begin{align}
	H(z)&=\sum_{n\in\mbb{Z}}H(n)z^{-n-1}, \\
	\Gamma_{\alpha}(z)&=e^{\alpha}z^{\alpha (0)}e^{-\sum_{j<0}\frac{z^{-j}}{j}\alpha (j)}e^{-\sum_{j>0}\frac{z^{-j}}{j}\alpha (j)},
\end{align}
respectively.
Then $V_{L}$ admits a unique structure of a vertex algebra.

Let $\{H_{i}\}_{i=1}^{\ell}$ be an orthonormal basis of $\mfrak{h}$ with respect to $(\cdot|\cdot)$.
Then the vector
\begin{equation}
	\omega=\frac{1}{2}\sum_{i=1}^{\ell}H_{i}(-1)^{2}\ket{0}\otimes e^{0}
\end{equation}
is a conformal vector of central charge $\ell$.

The irreducible $V_{L}$-modules are classified by elements of $L^{\ast}/L$.\cite{Dong1993}
Here $L^{\ast}$ is the dual lattice of $L$ in $\mfrak{h}$, then $L$ is naturally a sublattice of $L^{\ast}$.
For $\varpi\in L^{\ast}$, we can construct a $V_{L}$-module in the following way.
Let $\mbb{C}[L+\varpi]$ be a vector space spanned by elements of $L+\varpi$ so that $\mbb{C}[L+\varpi]=\bigoplus_{\beta\in L}e^{\beta+\varpi}$,
on which a Lie subalgebra $\mfrak{h}\otimes\mbb{C}[\zeta]\oplus\mbb{C}K$ of the Heisenberg algebra acts as
$H(m)e^{\beta+\varpi}=0$ for $m>0$ and $H(0)e^{\beta+\varpi}=(H|\beta+\varpi)e^{\beta+\varpi}$ for $H\in\mfrak{h}$ and $\beta\in L$, and $K=\mrm{Id}$.
Then the $V_{L}$-module $V_{L+\varpi}$ is constructed as
\begin{equation}
	V_{L+\varpi}=\mrm{Ind}_{\mfrak{h}\otimes\mbb{C}[\zeta]\oplus\mbb{C}K}^{\what{\mfrak{h}}}\mbb{C}[L+\varpi],
\end{equation}
on which the action of $V_{L}$ is defined in an obvious way.
It is also clear that $V_{L+\varpi}$ depends only on the equivalence class $[\varpi]$ of $\varpi$ in $L^{\ast}/L$.

\subsection{Frenkel-Kac construction}
\label{subsect:app_frenkel_kac}
One of the most significant examples of lattice vertex algebras is one associated with a root lattice of ADE type,
which is isomorphic to the irreducible affine vertex algebra associated with the corresponding Lie algebra.
We shall explain this example.

Let $\mfrak{g}$ be a finite-dimensional simple Lie algebra of ADE type and fix its Cartan subalgebra $\mfrak{h}$.
Correspondingly we denote the set of roots by $\Delta$, and the root lattice by $Q=\mbb{Z}\Delta$.
Let $(\cdot|\cdot)$ be the nondegenerate symmetric invariant bilinear form on $\mfrak{g}$
normalized so that $(\theta|\theta)=2$ for the highest root $\theta$.
Let $\Pi=\{\alpha_{1},\cdots,\alpha_{\ell}\}$ be the set of simple roots, then they form a basis for the root lattice.
We also denote the root space decomposition of $\mfrak{g}$ by $\mfrak{g}=\mfrak{h}\oplus\bigoplus_{\alpha\in \Delta}\mfrak{g}_{\alpha}$,
where $\mfrak{g}_{\alpha}=\mbb{C}E_{\alpha}$ is the root space of the root $\alpha\in\Delta$ spanned by normalized vector $E_{\alpha}$ so that
$(E_{\alpha}|E_{-\alpha})=1$,
and the set of simple coroots by $\Pi^{\vee}=\{\alpha_{1}^{\vee},\cdots,\alpha_{\ell}^{\vee}\}$.

\begin{thm}[Frenkel-Kac\cite{FrenkelKac1980}]
There is an isomorphism $L_{\mfrak{g},k}\to V_{Q}$ of vertex algebras such that
\begin{equation}
	\alpha_{i}^{\vee}(-1)\ket{0}\mapsto \alpha_{i}(-1)\ket{0},\ \ E_{\alpha}(-1)\ket{0}\mapsto e^{\alpha},\ \ \alpha\in\Delta.
\end{equation}
\end{thm}

%app_ito_process
\section{Ito process on a Lie group}
\label{sect:app_ito_lie_group}
This appendix is devoted to a short description of Ito processes on Lie groups.
Detailed expositions on this matter can be found in the literature.\cite{Chirikjian2012, Applebaum2014}

Let $G$ be a finite-dimensional complex Lie group and $\mfrak{g}$ be its Lie algebra.
A strategy to construct an Ito process on the Lie group $G$ may be exponentiating an Ito process on the Lie group $\mfrak{g}$.
For convenience of description, we take a basis $\{X_{i}\}_{i=1}^{\dim\mfrak{g}}$ of $\mfrak{g}$.
Then an Ito process $X_{t}$ on $\mfrak{g}$ expanded in this basis so that $X_{t}=\sum_{i=1}^{\dim\mfrak{g}}x^{i}_{t}X_{i}$,
where $x^{i}_{t}$ are Ito processes that are characterized by SDEs of the form
\begin{equation}
	dx^{i}_{t}=\overline{x}^{i}_{t}dt+\sum_{j\in I_{\mrm{B}}}x^{i}_{(j)t}dB_{t}^{(j)}.
\end{equation}
Here $\overline{x}^{i}_{t}$ and $x^{i}_{(j)t}$ are random processes with proper finiteness properties,
and $B_{t}^{(j)}$ are mutually independent Brownian motions labeled by a set $I_{\mrm{B}}$.
We set the variance of $B_{t}^{(j)}$ as $\kappa_{j}$.
Then we can obtain a random process $g_{t}$ on $G$ by exponentiating $X_{t}$ as $g_{t}=\exp (X_{t})$,
but it is not easy to formulate the SDE on $g_{t}$ due to noncommutativity in the Lie algebra $\mfrak{g}$.

Instead, we construct a random process on $G$ via the McKean-Gangolli injection.\cite{Mckean2005}
In this approach, we identify the value $X_{t}$ at each time $t$ as a left invariant vector field on $G$,
and a random process $g_{t}$ on $G$ evolves along this random vector field.
Then the infinitesimal time evolution of $g_{t}$ is described by
\begin{equation}
	\label{eq:McKean_Gangolli_injection}
	g_{t+dt}=g_{t}\exp\left(\sum_{i=1}^{\dim\mfrak{g}}dx^{i}_{t}X_{i}\right).
\end{equation}
To formulate the SDE on such a constructed $g_{t}$, we take finite-dimensional faithful representation $V$ of $\mfrak{g}$.
Then on the vector space $V$, an action of $G$ is defined by exponentiating the action of $\mfrak{g}$.
In the following, we do not distinguish an element of $\mfrak{g}$ from its action on $V$.
When we expand the exponential function in Eq. (\ref{eq:McKean_Gangolli_injection}) and notice that quadratic terms in $dx^{i}_{t}$ may give contributions proportional to $dt$,
we obtain an SDE
\begin{equation}
	\label{eq:sde_lie_group}
	g_{t}^{-1}dg_{t}
		=\left(\sum_{i=1}^{\dim\mfrak{g}}\overline{x}^{i}_{t}X_{i}+\frac{1}{2}\sum_{j\in I_{\mrm{B}}}\kappa_{j}\left(\sum_{i=1}^{\dim\mfrak{g}}x^{i}_{(j)t}X_{i}\right)^{2}\right)dt
			+\sum_{i=1}^{\dim\mfrak{g}}\sum_{j\in I_{\mrm{B}}}x^{i}_{(j)t}dB_{t}^{(j)}.
\end{equation}
We regard this equation as the standard form of an SDE on an Ito processes on a Lie group.

We have to handle a random process on an infinite-dimensional Lie group in application to SLE.
The construction above can be naturally extended to an infinite-dimensional setting.
Let $\mfrak{g}$ be an infinite-dimensional Lie algebra and $G$ be the corresponding Lie group.
Examples of such infinite-dimensional Lie groups include the group of coordinate transformations $\mrm{Aut}\mcal{O}$ on a formal disk, 
loop groups of finite-dimensional Lie groups and their semi-direct products.
In typical cases, a faithful representation $V$ of $\mfrak{g}$ is infinite dimensional,
thus it is, in general, nontrivial whether the action of $\mfrak{g}$ on $V$ is exponentiated to an action of $G$, but we assume that it is.
The validity of this assumption can be verified for each example.
We also assume that the infinite sum that appears in Eq.(\ref{eq:McKean_Gangolli_injection}) in the case of $\dim\mfrak{g}=\infty$ makes sense.
Then the McKean-Gangolli injection works to construct a random process on the Lie group $G$ from an Ito process on $\mfrak{g}$,
and an SDE of the form Eq.(\ref{eq:sde_lie_group}) characterizes the random process.

%app_derivation_SDE
\section{Derivation of SDEs}
\label{sect:app_SDE}
As a proof of Prop. \ref{prop:sde_internal_sl2},
we derive SDEs on $e_{t}(\zeta)$, $h_{t}(\zeta)$, and $f_{t}(\zeta)$ so that
the random process $\scr{G}_{t}=e^{{\bf e}_{t}}e^{{\bf h}_{t}}e^{{\bf f}_{t}}Q(\rho_{t})$ satisfies
\begin{equation}
	\scr{G}_{t}^{-1}d\scr{G}_{t}=
	\left(-2L_{-2}+\frac{\kappa}{2}L_{-1}^{2}+\frac{\tau}{2}\sum_{r=1}^{3}X_{r}(-1)^{2}\right)dt
	+L_{-1}dB_{t}^{(0)}-\sum_{r=1}^{3}X_{r}(-1)dB_{t}^{(r)}.
\end{equation}
Here $\{X_{r}\}_{r=1}^{3}$ is an orthonormal basis of $\mfrak{sl}_{2}$ defined by
\begin{equation}
	X_{1}=\frac{1}{\sqrt{2}}H,\ \ X_{2}=\frac{1}{\sqrt{2}}(E+F),\ \ X_{3}=\frac{i}{\sqrt{2}}(E-F),
\end{equation}
and $B_{t}^{(i)}$, $i=0,1,2,3$ are independent Brownian motions with variances given by
\begin{equation}
	dB_{t}^{(0)}\cdot dB_{t}^{(0)}=\kappa dt,\ \ dB_{t}^{(r)}\cdot dB_{t}^{(r)}=\tau dt,\ \ r=1,2,3.
\end{equation}
Since each element $X\otimes f(\zeta)$ in the affine Lie algebra transforms under adjoint action by $Q(\rho_{t})$ as
$Q(\rho_{t})^{-1}X\otimes f(\zeta) Q(\rho_{t})=X\otimes f(\rho^{-1}_{t}(\zeta))$,
it suffices to derive SDEs so that $\Theta_{t}=e^{{\bf e}_{t}}e^{{\bf h}_{t}}e^{{\bf f}_{t}}$ satisfies
\begin{equation}
	\Theta_{t}^{-1}d\Theta_{t}=\frac{\tau}{2}\sum_{r=1}^{3}(X_{r}\otimes \rho_{t}(\zeta)^{-1})^{2}dt-\sum_{r=1}^{3}X_{r}\otimes \rho_{t}(\zeta)^{-1} dB_{t}^{(r)}.
\end{equation}
We suppose that $e_{t}(\zeta)$, $h_{t}(\zeta)$, and $f_{t}(\zeta)$ satisfy
\begin{align}
	de_{t}(\zeta)&=\overline{e}_{t}(\zeta)dt +\sum_{r=1}^{3}e_{t}^{r}(\zeta)dB_{t}^{(r)}, \\
	dh_{t}(\zeta)&=\overline{h}_{t}(\zeta)dt +\sum_{r=1}^{3}h_{t}^{r}(\zeta)dB_{t}^{(r)}, \\
	df_{t}(\zeta)&=\overline{f}_{t}(\zeta)dt +\sum_{r=1}^{3}f_{t}^{r}(\zeta)dB_{t}^{(r)}.
\end{align}
Then by Ito calculus, we obtain
\begin{align}
	de^{{\bf e}_{t}}&=e^{{\bf e}_{t}}\left(E\otimes \overline{e}_{t}(\zeta)+\frac{\tau}{2}(E\otimes e_{t}^{r}(\zeta))^{2}\right)dt +e^{{\bf e}_{t}}\sum_{r=1}^{3}E\otimes e_{t}^{r}(\zeta) dB_{t}^{(r)}, \\
	de^{{\bf h}_{t}}&=e^{{\bf h}_{t}}\left(H\otimes \overline{h}_{t}(\zeta)+\frac{\tau}{2}(H\otimes h_{t}^{r}(\zeta))^{2}\right)dt +e^{{\bf h}_{t}}\sum_{r=1}^{3}H\otimes h_{t}^{r}(\zeta) dB_{t}^{(r)}, \\
	de^{{\bf f}_{t}}&=e^{{\bf f}_{t}}\left(F\otimes \overline{f}_{t}(\zeta)+\frac{\tau}{2}(F\otimes f_{t}^{r}(\zeta))^{2}\right)dt +e^{{\bf f}_{t}}\sum_{r=1}^{3}F\otimes f_{t}^{r}(\zeta) dB_{t}^{(r)}.
\end{align}
The increment of $\Theta_{t}$ is also computed as
\begin{align}
	d\Theta_{t}=
	&(de^{{\bf e}_{t}})e^{{\bf h}_{t}}e^{{\bf f}_{t}}+e^{{\bf e}_{t}}(de^{{\bf h}_{t}})e^{{\bf f}_{t}}+e^{{\bf e}_{t}}e^{{\bf h}_{t}}(de^{{\bf f}_{t}})  \notag \\
	&+(de^{{\bf e}_{t}})(de^{{\bf h}_{t}})e^{{\bf f}_{t}}+(d e^{{\bf e}_{t}})e^{{\bf h}_{t}}(de^{{\bf f}_{t}})+e^{{\bf e}_{t}}(de^{{\bf h}_{t}})(de^{{\bf f}_{t}}).
\end{align}
Terms in the increment $d\Theta_{t}$ proportional to increments of the Brownian motions are
\begin{align}
	\sum_{r=1}^{3}\Biggl( 
	&E\otimes e^{-2h_{t}(\zeta)}e_{t}^{r}(\zeta)  \notag \\
	&+H\otimes (e^{-2h_{t}(\zeta)}f_{t}(\zeta)e_{t}^{r}(\zeta)+h_{t}^{r}(\zeta)) \notag \\
	&+F\otimes (f_{t}^{r}(\zeta)-e^{-2h_{t}(\zeta)}f_{t}(\zeta)^{2}e_{t}^{r}(\zeta)-2f_{t}(\zeta)h_{t}^{r}(\zeta))\Biggr) dB_{t}^{(r)}
\end{align}
Comparing this to $\sum_{r=1}^{3}X_{r}\otimes \rho_{t}(\zeta)^{-1}dB_{t}^{(r)}$, 
we identify $e_{t}^{r}(\zeta)$, $h_{t}^{r}(\zeta)$ and $f_{t}^{r}(\zeta)$ as
\begin{align}
	e_{t}^{1}(\zeta)&=0, & h_{t}^{1}(\zeta)&=-\frac{1}{\sqrt{2}\rho_{t}(\zeta)}, & f_{t}^{1}(\zeta)&= -\frac{\sqrt{2}f_{t}(\zeta)}{\rho_{t}(\zeta)}, \\
	e_{t}^{2}(\zeta)&=-\frac{e^{2h_{t}(\zeta)}}{\sqrt{2}\rho_{t}(\zeta)}, & h_{t}^{2}(\zeta)&=\frac{f_{t}(\zeta)}{\sqrt{2}\rho_{t}(\zeta)}, & f_{t}^{2}(\zeta)&=-\frac{1-f_{t}(\zeta)^{2}}{\sqrt{2}\rho_{t}(\zeta)}, \\
	e_{t}^{3}(\zeta)&=-\frac{ie^{2h_{t}(\zeta)}}{\sqrt{2}\rho_{t}(\zeta)}, & h_{t}^{3}(\zeta)&=\frac{if_{t}(\zeta)}{\sqrt{2}\rho_{t}(\zeta)}, & f_{t}^{3}(\zeta)&=\frac{i(1+f_{t}(\zeta)^{2})}{\sqrt{2}\rho_{t}(\zeta)}.
\end{align}
Then the term in the increment $d\Theta_{t}$ proportional to $dt$ becomes
\begin{align}
	&E\otimes e^{-2h_{t}(\zeta)}\overline{e}_{t}(\zeta) \notag \\
	&+H\otimes \left(\overline{h}_{t}(\zeta)+e^{-2h_{t}(\zeta)}f_{t}(\zeta)\overline{e}_{t}(\zeta)+\frac{\tau}{2\rho_{t}(\zeta)^{2}}\right) \notag \\
	&+F\otimes \left(\overline{f}_{t}(\zeta)-e^{-2h_{t}(\zeta)}f_{t}(\zeta)^{2}\overline{e}_{t}(\zeta)-2f_{t}(\zeta)\overline{h}_{t}(\zeta)-\frac{\tau f_{t}(\zeta)}{\rho_{t}(\zeta)^{2}}\right) \notag \\
	&+\frac{\tau}{2}\sum_{r=1}^{3}(X_{r}\otimes \rho(\zeta)^{-1})^{2}.
\end{align}
Comparing this to $\frac{\tau}{2}\sum_{r=1}^{3}(X_{r}\otimes \rho_{t}(\zeta)^{-1})^{2}$, we obtain
\begin{equation}
	\overline{e}_{t}(\zeta)=0,\ \ 
	\overline{h}_{t}(\zeta)=-\frac{\tau}{2}\rho_{t}(\zeta)^{-2},\ \ 
	\overline{f}_{t}(\zeta)=0.
\end{equation}
We can finally formulate SDEs
\begin{align}
	de_{t}(\zeta)=
	&-\frac{e^{2h_{t}(\zeta)}}{\sqrt{2}\rho_{t}(\zeta)}dB_{t}^{(2)}-\frac{ie^{2h_{t}(\zeta)}}{\sqrt{2}\rho_{t}(\zeta)}dB_{t}^{(3)}, \\
	dh_{t}(\zeta)=
	&-\frac{\tau}{2}\rho_{t}(\zeta)^{-2}dt -\frac{1}{\sqrt{2}\rho_{t}(\zeta)}dB_{t}^{(1)}
	+\frac{f_{t}(\zeta)}{\sqrt{2}\rho_{t}(\zeta)}dB_{t}^{(2)}+\frac{if_{t}(\zeta)}{\sqrt{2}\rho_{t}(\zeta)}dB_{t}^{(3)}, \\
	df_{t}(\zeta)=
	&-\frac{\sqrt{2}f_{t}(\zeta)}{\rho_{t}(\zeta)}dB_{t}^{(1)}-\frac{1-f_{t}(\zeta)^{2}}{\sqrt{2}\rho_{t}(\zeta)}dB_{t}^{(2)}
	+\frac{i(1+f_{t}(\zeta)^{2})}{\sqrt{2}\rho_{t}(\zeta)}dB_{t}^{(3)}.
\end{align}

%app_derivation_operator
\section{Derivation of operators $\scr{X}_{\ell}$}
\label{sect:app_operator_X}
In this appendix, we derive the operators $\scr{X}_{\ell}$ in Sect.\ref{sect:affine_symmetry} that define an action of $\what{\mfrak{sl}}_{2}$
on a space of SLE local martingales.

We first derive differential equations satisfied by $\scr{G}=e^{\bf e}e^{\bf h}e^{\bf f}Q(g)$.
Here ${\bf e}=E\otimes e(\zeta)$, ${\bf h}=H\otimes h(\zeta)$ and ${\bf f}=F\otimes f(\zeta)$ are elements in $\mfrak{g}\otimes \mbb{C}[[\zeta^{-1}]]\zeta^{-1}$
with $e(\zeta)=\sum_{n<0}e_{n}\zeta^{n}$, $h(\zeta)=\sum_{n<0}h_{n}\zeta^{n}$, and $f(\zeta)=\sum_{n<0}f_{n}\zeta^{n}$.
The element $g\in\mrm{Aut}_{+}\mcal{O}$ is identified with a Laurant series $g(z)=z+\sum_{n\le 0}g_{n}z^{n}$.
By differentiating $\scr{G}$ by $e_{n}$ we obtain
\begin{equation}
	\frac{\del \scr{G}}{\del e_{n}}=e^{\bf e}E\otimes \zeta^{n} e^{\bf h}e^{\bf f}G(g).
\end{equation}
After transferring $E\otimes \zeta^{n}$ to the rightest position in the product, we have a differential equation
\begin{align}
	\scr{G}^{-1}\frac{\del \scr{G}}{\del e_{n}}=&E\otimes e^{-2h(g^{-1}(\zeta))}g^{-1}(\zeta)^{n}+H\otimes e^{-2h(g^{-1}(\zeta))}f(g^{-1}(\zeta))g^{-1}(\zeta)^{n} \notag \\
	&-F\otimes e^{-2h(g^{-1}(\zeta))}f(g^{-1}(\zeta))^{2}g^{-1}(\zeta)^{n}
\end{align}
Similarly, we can compute derivatives of $\scr{G}$ in variables $h_{n}$ and $f_{n}$ as
\begin{align}
	\scr{G}^{-1}\frac{\del \scr{G}}{\del h_{n}}&=H\otimes g^{-1}(\zeta)^{n}-2F\otimes f(g^{-1}(\zeta))g^{-1}(\zeta)^{n}, \\
	\scr{G}^{-1}\frac{\del \scr{G}}{\del f_{n}}&=F\otimes g^{-1}(\zeta)^{n}.
\end{align}
We shall invert these relations; namely, we express an object like $\scr{G}X\otimes \theta(\zeta)$
for a certain $\theta(\zeta)\in \mbb{C}[[\zeta^{-1}]]\zeta^{-1}$ by a linear combination of derivatives of $\scr{G}$.
\begin{lem}
\label{lem:differential_equation}
Let $\theta(\zeta)\in \mbb{C}[[\zeta^{-1}]]\zeta^{-1}$.
Then
\begin{align}
	\scr{G}F\otimes \theta(\zeta)&=\sum_{n\le -1}\left(\mrm{Res}_{w}w^{-n-1}\theta(g(w))\right)\frac{\del \scr{G}}{\del f_{n}}, \\
	\scr{G}H\otimes \theta(\zeta)=&\sum_{n\le -1}\left(\mrm{Res}_{w}w^{-n-1}\theta(g(w))\right)\frac{\del \scr{G}}{\del h_{n}} 
		+2\sum_{n\le -1}\left(\mrm{Res}_{w}w^{-n-1}f(w)\theta(g(w))\right)\frac{\del \scr{G}}{\del f_{n}}, \\
	\scr{G}E\otimes \theta(\zeta)=
	&\sum_{n\le -1}\left(\mrm{Res}_{w}w^{-n-1}e^{2h(w)}\theta(g(w))\right)\frac{\del \scr{G}}{\del e_{n}}
		-\sum_{n\le -1}\left(\mrm{Res}_{w}w^{-n-1}f(w)\theta(g(w))\right)\frac{\del \scr{G}}{\del h_{n}} \notag \\
	&-\sum_{n\le -1}\left(\mrm{Res}_{w}w^{-n-1}f(w)^{2}\theta(g(w))\right)\frac{\del \scr{G}}{\del f_{n}}.
\end{align}
\end{lem}

\begin{proof}
We have to search for an infinite series $a(z)=\sum_{n\le -1}a_{n}z^{n}$ such that $a(g^{-1}(\zeta))=\theta(\zeta)$ for a given infinite series $\theta(z)\in\mbb{C}[[\zeta^{-1}]]\zeta^{-1}$.
Such an infinite series is indeed obtained by setting $a_{n}=\mrm{Res}_{w}w^{-n-1}\theta(g(w))$,
which enables us to obtain the desired result.
\end{proof}

We next prepare formulas to compute $\scr{G}^{-1}X(-\ell)\scr{G}$ for $X\in\mfrak{sl}_{2}$ and $\ell\in\mbb{Z}$,
which is straightforward from the formulas in Subsect.\ref{subsect:formulae_sl2}.
\begin{lem}
\label{lem:adjoint_GinvXG}
We set $\xi:=g^{-1}(\zeta)$:
\begin{align}
	\scr{G}^{-1}E\otimes \zeta^{-\ell}\scr{G}
	=& E\otimes e^{-2h(\xi)}\xi^{-\ell}+H\otimes e^{-2h(\xi)}f(\xi)\xi^{-\ell}  \notag \\
		&-F\otimes e^{-2h(\xi)}f(\xi)^{2}\xi^{-\ell}-k\mrm{Res}_{w}\del f(w)e^{-2h(w)}w^{-\ell}, \\
	\scr{G}^{-1}H\otimes \zeta^{-\ell}\scr{G}
	=& 2E\otimes e^{-2h(\xi)}e(\xi)\xi^{-\ell} \notag \\
		&+H\otimes (1+2e^{-2h(\xi)}e(\xi)f(\xi))\xi^{-\ell} \notag \\
		&-2F\otimes( f(\xi)+ e^{-2h(\xi)}e(\xi)f(\xi)^{2})\xi^{-\ell} \notag \\
		&-2k\mrm{Res}_{w}(\del h(w)+\del f(w)e^{-2h(w)}e(w))w^{-\ell}, \\
	\scr{G}^{-1}F\otimes \zeta^{-\ell}\scr{G}
	=& -E\otimes e^{-2h(\xi)}e(\xi)^{2}\xi^{-\ell}  \notag \\
		&-H\otimes (e(\xi)+e^{-2h(\xi)}e(\xi)^{2}f(\xi))\xi^{-\ell} \notag \\
		&+F\otimes (e^{2h(\xi)}+2e(\xi)f(\xi)+e^{-2h(\xi)}e(\xi)^{2}f(\xi)^{2})\xi^{-\ell} \notag \\
		&+k\mrm{Res}_{w}(2\del h(w)e(w)-\del e(w) +\del f(w)e^{-2h(w)}e(w)^{2})w^{-\ell}.
\end{align}
\end{lem}

Next we express the objects like $\scr{G}X\otimes \theta(\zeta)\mcal{Y}(v, x)\ket{0}$ for $X\in\mfrak{sl}_{2}$, $\theta(\zeta)\in \mbb{C}((\zeta^{-1}))$
and an intertwining operator $\mcal{Y}(-,x)$ in a convenient form with the help of Lemma \ref{lem:differential_equation}.
\begin{lem}
\label{lem:divide_X}
Let $\mcal{Y}(-,x)$ be an intertwining operator, $v\in L(\Lambda)$ be a primary vector in the top space of $L_{\mfrak{sl}_{2}}(\Lambda,k)$,
and $\theta(\zeta)\in\mbb{C}((\zeta^{-1}))$.
Then
\begin{align}
	\scr{G}E\otimes \theta(\zeta)\mcal{Y}(v,x)\ket{0}=
	&\Biggl(\sum_{n\le -1}\mrm{Res}_{z}\mrm{Res}_{w}\frac{w^{-n-1}e^{2h(w)}\theta(z)}{g(w)-z}\frac{\del}{\del e_{n}}  \notag \\
		&\hspace{10pt}-\sum_{n\le -1}\mrm{Res}_{z}\mrm{Res}_{w}\frac{w^{-n-1}f(w)\theta(z)}{g(w)-z}\frac{\del}{\del h_{n}} \notag \\
		&\hspace{10pt} -\sum_{n\le -1}\mrm{Res}_{z}\mrm{Res}_{w}\frac{w^{-n-1}f(w)^{2}\theta(z)}{g(w)-z}\frac{\del}{\del f_{n}}\Biggr)\scr{G}\mcal{Y}(v,x)\ket{0} \notag \\
	&+\mrm{Res}_{z}\frac{\theta(z)}{z-x}\scr{G}\mcal{Y}(Ev,x)\ket{0}, \\
	\scr{G}H\otimes \theta(\zeta)\mcal{Y}(v,x)\ket{0}=
	&\Biggl(\sum_{n\le -1}\mrm{Res}_{z}\mrm{Res}_{w}\frac{w^{-n-1}\theta(z)}{g(w)-z}\frac{\del}{\del h_{n}}  \notag \\
	&\hspace{10pt}+2\sum_{n\le -1}\mrm{Res}_{z}\mrm{Res}_{w}\frac{w^{-n-1}f(w)\theta(z)}{g(w)-z}\frac{\del}{\del f_{n}}\Biggr)\scr{G}\mcal{Y}(v,x)\ket{0} \notag \\
	&+\mrm{Res}_{z}\frac{\theta(z)}{z-x}\scr{G}\mcal{Y}(Hv,x)\ket{0},  \\
	\label{eq:divide_F}
	\scr{G}F\otimes \theta(\zeta)\mcal{Y}(v,x)\ket{0}=
	&\sum_{n\le -1}\mrm{Res}_{z}\mrm{Res}_{w}\frac{w^{-n-1}\theta(z)}{g(w)-z}\frac{\del}{\del f_{n}}\scr{G}\mcal{Y}(v,x)\ket{0}  \notag \\
	&+\mrm{Res}_{z}\frac{\theta(z)}{z-x}\scr{G}\mcal{Y}(Fv,x)\ket{0}.
\end{align}
\end{lem}
\begin{proof}
As an example, we show Eq.(\ref{eq:divide_F}).
The other two equalities are shown in a similar way.
We divide a Leurant series $\theta(\zeta)=\sum_{n\in\mbb{Z}}\theta_{n}\zeta^{n}$ into the negative power part and the non-negative power part as
\begin{equation}
	\theta(\zeta)=\theta(\zeta)_{-}+\theta(\zeta)_{+},
\end{equation}
where $\theta(\zeta)_{-}=\sum_{n<0}\theta_{n}\zeta^{n}$ and $\theta(\zeta)_{+}=\sum_{n\ge 0}\theta_{n}\zeta^{n}$.
Notice that $\theta(\zeta)_{-}$ is expressed as the following:
\begin{equation}
	\theta(\zeta)_{-}=\mrm{Res}_{z}\frac{\theta(z)}{\zeta-z}.
\end{equation}
Together with Lemma \ref{lem:differential_equation}, this implies that
\begin{equation}
	\scr{G}F\otimes \theta(\zeta)_{-}=\sum_{n\le -1}\mrm{Res}_{z}\mrm{Res}_{w}\frac{w^{-n-1}\theta(z)}{g(w)-z}\frac{\del \scr{G}}{\del f_{n}}.
\end{equation}
Since $\mcal{Y}(v,x)$ is a primary field,
\begin{equation}
	[F(n),\mcal{Y}(v,x)]=x^{n}\mcal{Y}(Fv,x),
\end{equation}
which implies that
\begin{equation}
	[F\otimes \theta(\zeta)_{+},\mcal{Y}(v,x)]=\sum_{n=0}^{\infty}\theta_{n}x^{n}\mcal{Y}(Fv,x)=\mrm{Res}_{z}\frac{\theta(z)}{z-x}\mcal{Y}(Fv,x).
\end{equation}
Noting that $F\otimes \theta(\zeta)_{+}$ annihilates the vacuum vector $\ket{0}$, we obtain the desired result.
\end{proof}

For an intertwining operator $\mcal{Y}(-,z)$ of type $\binom{L_{\mfrak{sl}_{2}}(\Lambda.k)}{L_{\mfrak{sl}_{2}}(\Lambda,k),\ \ L_{\mfrak{sl}_{2},k}}$,
we regard $\braket{u|\scr{G}\mcal{Y}(-,x)|0}$ as an element of $L(\Lambda)^{\ast}[g_{n+1},e_{n},h_{n},f_{n}|n<0][[x]]$.
The dual space $L(\Lambda)^{\ast}$ is equipped with a representation $\pi$ of $\mfrak{sl}_{2}$ defined by
$(\pi(X)\phi)(v)=-\phi(Xv)$ for $X\in\mfrak{sl}_{2}$, $\phi\in L(\Lambda)^{\ast}$ and $v\in L(\Lambda)$.
Combining Lemma \ref{lem:adjoint_GinvXG} and \ref{lem:divide_X}, we derive operators $\scr{X}_{\ell}$ that satisfy
$\braket{X(\ell)u|\scr{G}\mcal{Y}(-,x)|0}=\scr{X}_{\ell}\braket{u|\scr{G}\mcal{Y}(-,x)|0}$ for $X\in\mfrak{sl}_{2}$ and $\ell\in\mbb{Z}$.

We begin with the computation of $\braket{E(\ell)u|\scr{G}\mcal{Y}(v,x)|0}$:
\begin{align}
	\braket{E(\ell)u|\scr{G}\mcal{Y}(v,x)|0}=-\braket{u|E(-\ell)\scr{G}\mcal{Y}(v,x)|0}=\scr{E}_{\ell}\braket{u|\scr{G}\mcal{Y}(v,x)|0},
\end{align}
where
\begin{align}
	\scr{E}_{\ell}=
	&-\sum_{n\le -1}\mrm{Res}_{z}\mrm{Res}_{w}\frac{w^{-n-1}e^{2h(w)}e^{-2h(z)}z^{-\ell}\pr{g}(z)}{g(w)-g(z)}\frac{\del}{\del e_{n}} \notag \\
	&-\sum_{n\le -1}\mrm{Res}_{z}\mrm{Res}_{w}\frac{w^{-n-1}e^{-2h(z)}(f(z)-f(w))z^{-\ell}\pr{g}(z)}{g(w)-g(z)}\frac{\del}{\del h_{n}} \notag \\
	&+\sum_{n\le -1}\mrm{Res}_{z}\mrm{Res}_{w}\frac{w^{-n-1}e^{-2h(z)}(f(z)-f(w))^{2}z^{-\ell}\pr{g}(z)}{g(w)-g(z)}\frac{\del}{\del f_{n}} \notag \\
	&+\mrm{Res}_{z}\frac{e^{-2h(z)}z^{-\ell}\pr{g}(z)}{g(z)-x}\pi(E) \notag \\
	&+\mrm{Res}_{z}\frac{e^{-2h(z)}f(z)z^{-\ell}\pr{g}(z)}{g(z)-x}\pi(H) \notag \\
	&-\mrm{Res}_{z}\frac{e^{-2h(z)}f(z)^{2}z^{-\ell}\pr{g}(z)}{g(z)-x}\pi(F) \notag \\
	&+k\mrm{Res}_{z}\del f(z)e^{-2h(z)}z^{-\ell}.
\end{align}

We also obtain
\begin{align}
	\braket{H(\ell)u|\scr{G}\mcal{Y}(v,x)|0}=-\braket{u|H(-\ell)\scr{G}\mcal{Y}(v,x)|0}=\scr{H}_{\ell}\braket{u|\scr{G}\mcal{Y}(v,x)|0},
\end{align}
where
\begin{align}
	\scr{H}_{\ell}=
	&-2\sum_{n\le -1}\mrm{Res}_{z}\mrm{Res}_{w}\frac{w^{-n-1}e^{2h(w)}e^{-2h(z)}e(z)z^{-\ell}\pr{g}(z)}{g(w)-g(z)}\frac{\del}{\del e_{n}} \notag \\
	&-\sum_{n\le -1}\mrm{Res}_{z}\mrm{Res}_{w}\frac{w^{-n-1}(1+2e^{-2h(z)}(f(z)-f(w)))z^{-\ell}\pr{g}(z)}{g(w)-g(z)}\frac{\del}{\del h_{n}} \notag \\
	&-2\sum_{n\le -1}\mrm{Res}_{z}\mrm{Res}_{w}\frac{w^{-n-1}(f(w)-f(z)-e^{-2h(z)}e(z)(f(w)-f(z))^{2})z^{-\ell}\pr{g}(z)}{g(w)-g(z)}\frac{\del}{\del f_{n}} \notag \\
	&+2\mrm{Res}_{z}\frac{e^{-2h(z)}e(z)z^{-\ell}\pr{g}(z)}{g(z)-x}\pi(E) \notag \\
	&+\mrm{Res}_{z}\frac{(1+2e^{-2h(z)}e(z)f(z))z^{-\ell}\pr{g}(z)}{g(z)-x}\pi(H) \notag \\
	&-2\mrm{Res}_{z}\frac{(1+e^{-2h(z)}e(z)f(z))f(z)z^{-\ell}\pr{g}(z)}{g(z)-x}\pi(F) \notag \\
	&+2k\mrm{Res}_{z}(\del h(z)-\del f(z)e^{-2h(z)}e(z))z^{-\ell},
\end{align}
and
\begin{equation}
	\braket{F(\ell)u|\scr{G}\mcal{Y}(v,x)|0}=-\braket{u|F(-\ell)\scr{G}\mcal{Y}(v,x)|0}=\scr{F}_{\ell}\braket{u|\scr{G}\mcal{Y}(v,x)|0},
\end{equation}
where
\begin{align}
	\scr{F}_{\ell}=
	&\sum_{n\le -1}\mrm{Res}_{z}\mrm{Res}_{w}\frac{w^{-n-1}e^{2h(w)}e^{-2h(z)}e(z)^{2}z^{-\ell}\pr{g}(z)}{g(w)-g(z)}\frac{\del}{\del e_{n}} \notag \\
	&-\sum_{n\le -1}\mrm{Res}_{z}\mrm{Res}_{w}\frac{w^{-n-1}(1+e^{-2h(z)}e(z)(f(w)-f(z)))e(z)z^{-\ell}\pr{g}(z)}{g(w)-g(z)}\frac{\del}{\del h_{n}} \notag \\
	&-\sum_{n\le -1}\mrm{Res}_{z}\mrm{Res}_{w}w^{-n-1}\Biggl[\frac{e^{2h(z)}+2e(z)(f(z)-f(w))}{g(w)-g(z)} \notag \\
	&\hspace{120pt}+\frac{e^{-2h(z)}e(z)^{2}(f(z)-f(w))^{2}}{g(w)-g(z)}\Biggr]z^{-\ell}\pr{g}(z)\frac{\del}{\del f_{n}} \notag \\
	&-\mrm{Res}_{z}\frac{e^{-2h(z)}e(z)^{2}z^{-\ell}\pr{g}(z)}{g(z)-x}\pi(E) \notag \\
	&-\mrm{Res}_{z}\frac{(1+e^{-2h(z)}e(z)f(z))e(z)z^{-\ell}\pr{g}(z)}{g(z)-x}\pi(H) \notag \\
	&+\mrm{Res}_{z}\frac{(e^{2h(z)}+2e(z)f(z)+e^{-2h(z)}e(z)^{2}f(z)^{2})z^{-\ell}\pr{g}(z)}{g(z)-x}\pi (F) \notag \\
	&-\mrm{Res}_{z}(2\del h(z)e(z)-\del e(z)+\del f(z)e^{-2h(z)}e(z)^{2})z^{-\ell}.
\end{align}

We look for an operator $\scr{L}_{\ell}$ such that $\braket{L_{\ell}u|\scr{G}\mcal{Y}(v,x)|0}=\scr{L}_{\ell}\braket{u|\scr{G}\mcal{Y}(v,x)|0}$.
We first prepare a lemma.
\begin{lem}We set $\xi=g^{-1}(\zeta)$:
\begin{align}
	\scr{G}^{-1}L_{-\ell}\scr{G}
	=&\sum_{m\in\mbb{Z}}\left(\mrm{Res}_{z}z^{-\ell+1}g(z)^{-n-2}\pr{g}(z)^{2}\right)L_{m} \notag \\
	&-E\otimes e^{-2h(\xi)}\del e(\xi)\xi^{-\ell+1} \notag \\
	&-H\otimes (\del h(\xi)+e^{-2h(\xi)}f(\xi)\del e(\xi))\xi^{-\ell+1} \notag \\
	&-F\otimes (\del f(\xi)-2f(\xi)\del h(\xi)-e^{-2h(\xi)}f(\xi)^{2}\del e(\xi))\xi^{-\ell+1} \notag \\
	&+\mrm{Res}_{z}z^{-\ell+1}\left(\frac{c}{12}(Sg)(z)+k(\del h(z)^{2}+e^{-2h(z)}\del f(z)\del e(z))\right).
\end{align}
\end{lem}

Notice that $\scr{G}$ satisfies the same differential equation as the one in the case of the Virasoro algebra,\cite{Kytola2007}
thus
\begin{equation}
	\scr{G}L_{m}=-\sum_{n\le 0}\left(\mrm{Res}_{z}z^{-n-1}g(z)^{m+1}\right)\frac{\del \scr{G}}{\del g_{n}}
\end{equation}
for $m\le -1$.
Terms of type $\scr{G}X\otimes x(\zeta)$ for $X\in\mfrak{sl}_{2}$ can be also expressed as derivatives of $\scr{G}$ as shown previously.
Thus the desired operator $\scr{L}_{\ell}$ is specified as
\begin{align}
	\scr{L}_{\ell}=
	&-\sum_{n\le 0}\mrm{Res}_{z}\mrm{Res}_{w}\frac{z^{-\ell+1}w^{-n-1}\pr{g}(z)^{2}}{g(w)-g(z)}\frac{\del}{\del g_{n}} \notag \\
	&-\sum_{n\le -1}\mrm{Res}_{z}\mrm{Res}_{w}\frac{z^{-\ell+1}w^{-n-1}e^{2h(w)}e^{-2h(z)}\del e(z)\pr{g}(z)}{g(w)-g(z)}\frac{\del}{\del e_{n}} \notag \\
	&-\sum_{n\le -1}\mrm{Res}_{z}\mrm{Res}_{w}\frac{z^{-\ell+1}w^{-n-1}(\del h(z)+e^{-2h(z)}\del e(z)(f(z)-f(w)))\pr{g}(z)}{g(w)-g(z)}\frac{\del}{\del h_{n}} \notag \\
	&-\sum_{n\le -1}\mrm{Res}_{z}\mrm{Res}_{w}z^{-\ell+1}w^{-n-1}\Biggl[\frac{\del f(z)-2\del h(z)(f(z)-f(w))}{g(w)-g(z)} \notag \\
	&\hspace{150pt}-\frac{e^{-2h(z)}\del e(z)(f(z)-f(w))^{2}}{g(w)-g(z)}\Biggr]\pr{g}(z)\frac{\del}{\del f_{n}} \notag \\
	&+\mrm{Res}_{z}z^{-\ell+1}\pr{g}(z)^{2}\left(\frac{h}{(g(z)-x)^{2}}+\frac{1}{g(z)-x}\frac{\del}{\del x}\right) \notag \\
	&+\mrm{Res}_{z}\frac{z^{-\ell+1}e^{-2h(z)}\del e(z)\pr{g}(z)}{g(z)-x}\pi(E) \notag \\
	&+\mrm{Res}_{z}\frac{z^{-\ell+1}(\del h(z)+e^{-2h(z)}f(z)\del e(z))\pr{g}(z)}{g(z)-x}\pi(H) \notag \\
	&+\mrm{Res}_{z}\frac{z^{-\ell+1}(\del f(z)-2f(z)\del h(z)-e^{-2h(z)}f(z)^{2}\del e(z))\pr{g}(z)}{g(z)-x}\pi(F) \notag \\
	&+\mrm{Res}_{z}z^{-\ell+1}\left(\frac{c}{12}(Sg)(z)+k(\del h(z)^{2}+e^{-2h(z)}\del f(z)\del e(z))\right).
\end{align}

\addcontentsline{toc}{chapter}{Bibliography}
\bibliographystyle{alpha}
\bibliography{sle_cft}
\end{document}